\documentclass[lettersize,journal]{IEEEtran}
\usepackage{amsmath,amssymb,amsfonts,amsthm}
\usepackage{algorithmic}
\usepackage{array}
\usepackage[caption=false,font=normalsize,labelfont=sf,textfont=sf]{subfig}
\usepackage{textcomp}
\usepackage{stfloats}
\usepackage{url}
\usepackage{verbatim}
\usepackage{graphicx}
\usepackage{cite}
\usepackage{xcolor}
\usepackage{diagbox}
\usepackage{slashbox}
\usepackage{diagbox}
\usepackage{siunitx}
\usepackage[colorlinks,
            linkcolor=blue,     
            anchorcolor=blue, 
            citecolor=blue,       
            ]{hyperref}
\usepackage[caption=false,font=normalsize,labelfont=sf,textfont=sf]{subfig}
\usepackage{multirow}
\usepackage{tikz}
\usepackage{tikz-3dplot}
\usepackage{svg}
\usepackage{overpic}
\usetikzlibrary {arrows.meta,bending}

\usepackage[linesnumbered,lined,boxed,commentsnumbered,ruled,vlined]{algorithm2e}

\usepackage{pgfplots}
\pgfplotsset{compat=newest}

\newtheorem{theorem}{Theorem}
\newtheorem{proposition}{Proposition}
\newtheorem{lemma}{Lemma}

\usepackage{xifthen}
\newcommand{\prob}[1][]{
\ifthenelse{\isempty{#1}}%
      {\ensuremath{P}}%
    {\ensuremath{P\left\(#1\right\)}}%
}
\newcommand{\vect}[1]{\boldsymbol{\mathrm{#1}}}
\newcommand{\mat}[1]{\boldsymbol{\mathrm{#1}}}

\newcommand{\tr}{\mathrm{tr}}
\newcommand{\E}{\mathbb{E}}
\newcommand{\diag}{\mathrm{diag}}

\newcommand{\norm}[1]{\left\lVert#1\right\rVert}

\usepackage[acronym]{glossaries}
\makeglossaries
\glsdisablehyper
\newacronym{ml}{ML}{maximum likelihood}
\newacronym{gps}{GPs}{Gaussian processes}
\newacronym{gp}{GP}{Gaussian process}
\newacronym{mlp}{MLP}{multilayer perceptrons}
\newacronym{vlp}{VLP}{visible light positioning}
\newacronym{led}{LED}{light-emitting diode}
\newacronym{mse}{MSE}{mean squared error}
\newacronym{crlb}{CRLB}{Cramer-Rao lower bound}
\newacronym{fim}{FIM}{Fisher Information matrix}
\newacronym{rss}{RSS}{received signal strength}
\newacronym{se}{SE}{squared error}
\newacronym{snr}{SNR}{Signal-to-Noise Ratio}
\newacronym{cdf}{CDF}{cumulative distribution function}
\newacronym{ls}{LS}{least squares}

\begin{document}

\title{Enhancing RSS-Based Visible Light Positioning by Optimal Calibrating the LED Tilt and  Gain}



\author{Fan Wu, Nobby Stevens, Lieven De Strycker and  François Rottenberg,~\IEEEmembership{Member,~IEEE,}
\thanks{The authors are with the ESAT-DRAMCO, KU Leuven (Ghent), 9000 Ghent, Belgium (e-mail: fan.wu@kuleuven.be).}

\thanks{ Acknowledge the Chinese Scholarship Council (CSC) for the Ph.D
grant of Fan Wu (No. 202106340043)

This work has been submitted to the IEEE for possible publication.  Copyright may be transferred without notice, after which this version may no longer be accessible.

}}

\markboth{Journal of \LaTeX\ Class Files,~Vol.~14, No.~8, August~2021}%
{Shell \MakeLowercase{\textit{et al.}}: A Sample Article Using IEEEtran.cls for IEEE Journals}


\maketitle

\begin{abstract}
This paper presents an optimal calibration scheme and a weighted \gls{ls} localization algorithm for \gls{rss} based  \gls{vlp} systems, focusing on the often-overlooked impact of \gls{led} tilt. By optimally  calibrating \gls{led} tilt and gain, we significantly enhance \gls{vlp} localization accuracy. 
Our algorithm outperforms both machine learning \gls{gps} and traditional multilateration techniques. Against \gls{gps}, it achieves improvements of 58\% and 74\% in the 50th and 99th percentiles, respectively. When compared to multilateration, it reduces the 50th percentile error from 7.4~cm to 3.2~cm and the 99th percentile error from 25.7~cm to 11~cm. 
We introduce a low-complexity estimator for tilt and gain that meets the \gls{crlb} for the \gls{mse}, emphasizing its precision and efficiency. Further, we elaborate on optimal calibration measurement placement and refine the observation model to include residual calibration errors, thereby improving localization performance. The weighted \gls{ls} algorithm's effectiveness is validated through simulations and real-world data, consistently outperforming \gls{gps} and multilateration,  across various training set sizes and reducing outlier errors. Our findings underscore the critical role of \gls{led} tilt calibration in advancing \gls{vlp} system accuracy and contribute to a more precise model for indoor positioning technologies.
\end{abstract}


\glsreset{crlb}
\glsreset{ls}
\glsreset{mse}

\section{Introduction}
\IEEEPARstart{I}{n} recent years, the evolution of indoor positioning technologies has attracted substantial attention, particularly in overcoming the limitations of global positioning system in indoor environments \cite{shit2018location}. Among the myriad of emerging solutions, \gls{vlp} systems, primarily utilizing \gls{led}, have emerged as a promising solution \cite{chen2018research}.  Leveraging the widespread availability of LED-based lighting infrastructures, \gls{vlp} offers advantages such as high accuracy, cost-effectiveness, and minimal deployment complexity \cite{zhuang2018survey}. Employing techniques such as \gls{rss} \cite{zhuang2019low}, angle of arrival (AOA) \cite{yang2014three}, and time difference of arrival (TDOA) \cite{do2014tdoa}, \gls{vlp} has been shown to achieve remarkable precision in indoor localization. Among these techniques, \gls{rss} plays a crucial role, deducing location information from the intensity of light received at the photodiode (PD), thus offering a streamlined and cost-effective solution that eliminates the need for synchronized hardware components.

RSS-based indoor localization can be categorized into data-driven and model-based approaches.  Data-driven methods utilize machine learning algorithms like \gls{gps} \cite{knudde2020data,garbuglia2022bayesian,guan2021measuring}, \gls{mlp} \cite{arfaoui2021invoking,tran2022machine,guan2017high,yuan2018tilt}, and fingerprint techniques \cite{alam2018accurate,xu2023indoor,zhu2023centimeter}, require extensive data collection, including accurate PD coordinates and \gls{rss} values. However, they encounter challenges in terms of scalability and adaptability, especially when integrating new LEDs or adapting to changes in LED luminance \cite{royer2014lumen}.  Conversely, model-based approaches require a relatively small dataset to calibrate the attenuation model, which correlates RSS with the distances from the PD to LEDs, and employ multilateration \cite{kupper2005location} for estimating PD coordinates. While more streamlined in process, these methods are susceptible to model inaccuracies like LED or PD tilt and signal noise in line-of-sight (LOS) environments, which can compromise localization accuracy. A seemingly minor LED tilt of just \SI{3}{\degree}, for example, can lead to localization errors as large as \SI{30}{\cm} in our simulation  with  \SI{4}{\m} height of the LED. This underscores the critical importance of precise calibration and the consideration of tilt effects to maintain and enhance localization precision.

Extensive research has been devoted to advancing model-based localization in RSS-based \gls{vlp} systems. Studies like \cite{keskin2015comparative,steendam2016theoretical,zhang2014theoretical,amini2016theoretical} delve into the \gls{crlb} for distance estimation via TOA and RSS, while others \cite{zhang2014theoretical} offer theoretical insights into \gls{crlb} variations across different LED layouts. The development of custom-designed LED arrays is also explored \cite{cao2023led,zhou2019joint,hong2020angle,zhou2019performance,taparugssanagorn2014miso}, along with an examination of noise impacts on system performance \cite{amini2020theoretical}.  However, these approaches do not consider the effects of LED tilting. Studies \cite{stevens2020planar,stevens2018influence}  investigated the effects of tilt angles and deviations from the ideal Lambertian model of single LEDs on the path loss model, showing that even minor tilting can introduce single-LED distance estimation errors more than \SI{10}{\cm} \cite{stevens2019bias}. While these studies highlight the importance of accounting for LED tilt, they do not  propose new localization algorithms that explicitly consider this effect, nor do they identify the most efficient calibration scheme for practical use. This leaves room for more theoretical analysis and the development of better calibration scheme in VLP systems.

The key contributions of this paper, distinguishing it from prior studies, are:

\begin{enumerate}
    \item Development of a  calibration framework for LED tilt and  gain, backed by comprehensive theoretical analysis.
    \item A low complexity closed-form estimator for tilt and gain,  based on calibration measurements, which  achieves the \gls{crlb}, underscoring its efficiency.
    \item We demonstrate that the optimal location of calibration measurements can be done on a circle on the ground, uniformly distributed in angle and with a radius of approximately $0.55$ times the height of the LED. This not only improves the performance but also reduces the complexity of the tilt-gain estimator by simplifying the inversion process.
    \item We refine the observation model during the localization phase, taking into account residual calibration errors and formulate the   weighted \gls{ls} estimator. We also derive the \gls{crlb} on localization performance,  incorporating both additive noise and residual calibration noise.
    \item We extensively validate the benefits  of the proposed  weighted \gls{ls} technique both based on simulation and  experimental data, demonstrating substantial improvements in accuracy over \gls{gp} benchmarks and the traditional multilateration technique. Specifically, it reduced \gls{gp} errors from \SI{7.9}{\cm} to \SI{3.2}{\cm}  at P50 and from \SI{42.8}{\cm}  to \SI{11.1}{\cm}  at the P99, corresponding to improvement rates of 58\% and 74\%, respectively.
\end{enumerate}

The paper is structured as follows: Section~\ref{sec:SYSTEM DESCRIPTION} introduces the \gls{vlp} system. Section~\ref{sec:Calibration} presents the calibration model's closed-form solution and optimal calibration scheme. Section~\ref{sec:LOCALIZATION} links the calibration error with the   weighted \gls{ls} algorithm for improved localization. Section~\ref{sec:SIMULATION AND EXPERIMENTAL RESULTS} validates the model through simulations and practical application in a measurement dataset. The paper concludes with Section \ref{sec:CONCLUSION}.

\textit{Notations: }
Vectors and matrices are represented by bold lowercase and uppercase letters, respectively. The superscript~${}^T$ signifies transpose. The $\operatorname{tr}[\cdot]$, $\E[\cdot]$ and  $\Im(\cdot)$   indicate the trace, expectation and  imaginary  parts respectively. The notation $\mathcal{N}(\mu,\sigma^2)$ denotes a normal distribution  with mean $\mu$ and variance $\sigma^2$. The operator $\diag[\cdot]$ transforms a vector into a diagonal matrix, placing the elements of the vector along the matrix's diagonal. An identity matrix of size $K$ is $\mat{I}_K$. The operator $\mat{I}(\cdot)$ denotes the \gls{fim} of a vector, whose inverse yields the \gls{crlb}, represented by $\mat{C}(\cdot)$. The $(i,j)$-th element of a matrix $\mat{A}$ is indicated by $[\mat{A}]_{i,j}$. The notation $f(x)=O(g(x))$ as $x \rightarrow a$ indicates that positive numbers $\delta$ and $\lambda$ exist, ensuring $|f(x)| \leq \lambda g(x)$ for $0<|x-a|<\delta$.

\section{System Description}
\label{sec:SYSTEM DESCRIPTION}
In this study, we consider an RSS-based \gls{vlp} system that features multiple LEDs mounted on the ceiling, coupled with a receiver PD. As illustrated in Fig.~\ref{fig:led_tilt}, the calibration process is conducted individually for each LED. Therefore, without loss of generality and for the sake of clarity, we focus on a specific LED located at coordinates $\mathbf{r}_{S}$. This simplifies the notations of the calibration phase which is conducted independently per LED in Section~\ref{sec:Calibration}. In Section~\ref{sec:LOCALIZATION}, we consider localization based on all LEDs and we will generalize the notations by specifying  the LED index. To collect RSS values from each LED, a single PD is maneuvered to various locations $\mathbf{r}_{R,n}$, where $n=0,1,\ldots,N-1$, and $N$ denotes the total number of calibration points. These locations, $\vect{r}_{R,n}$, are assumed to be perfectly known during the calibration process.

In Fig.~\ref{fig:led_tilt}, $\mathbf{n}_{S}$ represents the unit normal vector (i.e., $\norm{\vect{n}_S}=1$) of the LED, while $\Delta \theta$ and $\Delta \gamma$ denote the  polar and azimuthal angles of $\mathbf{n}_{S}$, respectively.  We define $d_n$ as the distance between the PD and LED. This distance $d_n$ can be represented as the magnitude of the vector $\vect{d}_n=\mathbf{r}_{R, n}-\mathbf{r}_S$. Therefore, the $\cos (\phi_{S,n})$ and $\cos (\phi_{R,n})$ are defined as
\begin{equation}
\label{eq:phis}
\cos (\phi_{S,n}) = \mathbf{n}_{S}^T \cdot \frac{\mathbf{d}_n}{d_n}, \quad \cos (\phi_{R,n}) = \mathbf{n}_R^T \cdot \frac{-\mathbf{d}_n}{d_n}.
\end{equation}

\begin{figure}[!t]
    \centering
    \includesvg[width=0.6\linewidth]{Images/LED_tilt.svg}
    \caption{The coordinates of the transmitting  LED is $\mathbf{r}_{S}$, and the  receiving PD at the $n\text{-}th$ coordinates $\mathbf{r}_{R,n}$. (This picture is modified from \cite{stevens2020planar})}
    \label{fig:led_tilt}
\end{figure}

The RSS value $s_n$, received by the PD is
\begin{equation}
\begin{aligned}
\label{eq:channelmodel}
    s_n &= R_p P \frac{(m+1)A_r}{2 \pi d^2_n} \cos^{m} (\phi_{S,n}) \cos (\phi_{R,n}) + w_n.
\end{aligned}
\end{equation}

Here, $A_r$ represents the surface area of the PD,  $R_p$ is the PD's responsivity, $m$ is the Lambertian order of the LED, and $P$ is the transmitted optical power. The noise component $w_n$ is primarily composed of shot noise and Johnson noise. While shot noise originates from a Poisson process, it can be approximated as a Gaussian distribution for larger numbers, such as a large photon count, as suggested by \cite{554222}.  We thus model $w_n$ as an independently and identically distributed (i.i.d.) Gaussian noise, expressed as $w_n \sim \mathcal{N}(0, \sigma^2)$, where $\sigma^2$ is the noise variance.

\section{Calibration Methodology}
\label{sec:Calibration}

\subsection{Closed-Form Solution for Calibration}

During calibration, we position several points on the ground to facilitate measurements, ensuring that the PD remains horizontally aligned, with all samples set at a uniform height $h_n=h$. This method simplifies the calibration model, making $\cos( \phi_{R, n})$  equal to $h / d_n$, which not only streamlines the calibration process but also significantly improves measurement accuracy by providing a stable reference for calculation. A key objective of this calibration phase is to accurately estimate the tilt of LEDs mounted on the ceiling. Given the complexity of installing LEDs without any tilt due to their elevated positioning, our approach provides a practical solution by enabling the estimation and correction of any LED tilt through calibration.  This precision is critical, as even a minor tilt in the LEDs can dramatically affect the system's performance, leading to significant inaccuracies in localization. The process is refined by defining the gain as a constant scalar $c=$ $R_p P \frac{(m+1) A_r}{2 \pi}$, and adopting a Lambertian order of $m=1$ as recommended by existing literature \cite{raes2021usage,shi2019accuracy,li2024indoor,zhu2024efficient}. Under these assumptions,   (\ref{eq:channelmodel}) simplifies to
\begin{equation}
\label{eq:rangemodel}
    s_n = \frac{ch}{d_n^{4}} \mathbf{n}_S^T \vect{d}_n + w_n
\end{equation}
where $c$ is considered an unknown parameter to improve algorithm robustness, removing the necessity for exactly values of  $R_p$, $P$ and $A_r$.

Based on   (\ref{eq:rangemodel}), the calibration objective for $N$ observational data points is formulated as a \gls{ls} problem with unknown tilt $\vect{n}_S$ and gain $c$
\begin{equation}
\label{eq:obj1}
    \min\limits _{c, \mathbf{n}_S,\|\mathbf{n}_S\|=1} \frac{1}{2} \sum_{n=0}^{N-1}(s_n-\frac{ch}{d_n^{4}} \mathbf{n}_S^T \vect{d}_n)^2.
\end{equation}

The constraint $ \|\mathbf{n}_S  \|=1$  poses a challenge in terms of solvability. To solve it, we use a change of variable, converting the constrained optimization problem into a conventional linear least squares problem, with a simple closed-form solution.

Initially, we define the vector $\mathbf{g}_n = \frac{h}{d_n^4} \vect{d}_n \in \mathbb{R}^{3 \times 1}$, and denote
$$
\mathbf{s}= \left(\begin{array}{c}
s_0 \\
\vdots \\
s_{N-1}
\end{array}  \right) \in \mathbb{R}^{N \times 1}, \mathbf{G}= (\mathbf{g}_0, ..., \mathbf{g}_{N-1}) \in \mathbb{R}^{3 \times N}.
$$

Subsequently, we define $\mathbf{c} = c\mathbf{n}_S$. Hence   (\ref{eq:rangemodel}) becomes
\begin{equation}
\label{eq:newrangemodel}
\mathbf{s} = \mathbf{G}^T \mathbf{c} + \mathbf{w}
\end{equation}
where $\mathbf{w} = [w_0,.... w_{N-1}]^T$. The objective function (\ref{eq:obj1}) thus becomes
\begin{equation}
\label{eq:obj3}
\min _{\mathbf{c}} \frac{1}{2} \|\mathbf{s}-\mathbf{G}^T \mathbf{c}  \|^2.
\end{equation}

This formulation results in a conventional least squares problem, which is equivalent to the \gls{ml} estimator given that noise $\vect{w}$ is assumed i.i.d Gaussian. The solution can be derived using the pseudo-inverse \cite{kay1993fundamentals}
\begin{equation}
\label{eq:answer}
\hat{\mathbf{c}} =(\mathbf{G G}^T)^{-1} \mathbf{G} \mathbf{s} 
\end{equation}
and the unit normal vector and  gain can be retrieved easily as 
\begin{equation}
\label{eq:estimate_n_c}
    \hat{\vect{n}}_S =\frac{(\mat{G G}^T)^{-1} \mat{G} \mat{s}}{\norm{(\mat{G G}^T)^{-1} \mat{G} \mat{s}}}, 
\hat{c} =\norm{(\mat{G G}^T)^{-1} \mat{G} \mat{s}}.
\end{equation}

We can get an optimal and simple estimation result (\ref{eq:estimate_n_c}) when $m=1$. In the following section, the optimal calibration scheme is also based on $m=1$, but when $m$ differs, there might be a gap from optimality. However, the calibration scheme we propose can still serve as a useful guideline and the result can be calculated by numerical optimization methods.

After obtaining the estimated parameter $\hat{\vect{c}}$, the next step typically involves estimating the noise variance $\sigma^2$. The log-likelihood function is 
$$
\ln \mathcal{L}(\vect{s}; \vect{c}, \sigma^2) = -\frac{N}{2} \ln (2\pi \sigma^2) - \frac{1}{2\sigma^2} \norm{ \vect{s} - \mat{G}^T\vect{c}}^2.
$$

By differentiating this log-likelihood function with respect to $\sigma^2$ to find the \gls{ml} estimator of $\sigma^2$, we have
\begin{equation}
\label{eq:noise_estimate}
\hat{\sigma}^2  = \frac{1}{N} \norm{ \vect{s} - \mat{G}^T \vect{c}}^2
\end{equation}
where the $\vect{c}$ is given by its \gls{ml} estimate $\hat{\vect{c}}$.

\subsection{ Theoretical Analysis of Calibration Parameters}
To theoretically evaluate the precision of estimating parameters $\mathbf{c}$, $\mathbf{n}_S$ as well as $c$, it is essential to compute their respective \gls{mse} and we can also evaluate the \gls{crlb} to see their performance gap from optimality. The \gls{crlb} provides a lower bound on the variance of unbiased estimators, offering insight into the best possible accuracy achievable in parameter estimation.

\begin{proposition}
\label{proposition:crlb bold_c}
The estimator $\hat{\mathbf{c}}$ is unbiased and efficient, i.e., its \gls{mse} reaches the \gls{crlb}. Furthermore, the covariance matrix of the estimator is given by
$$
\E \left[(\vect{c}-\hat{\vect{c}})(\vect{c}-\hat{\vect{c}})^T \right] = \sigma^2 (\mat{GG}^T)^{-1}
$$
and the sum \gls{mse} is 
\begin{equation}
\label{eq:mini_summse}
    \E \left[\norm{\vect{c} - \hat{\vect{c}}}^2_2 \right]  = \sigma^2 \tr [(\mat{GG}^T)^{-1} ] .
\end{equation}

\end{proposition}

\begin{proof}
Please see Appendix.~\ref{ap:2}.
\end{proof}

\begin{figure}[!t]
    \centering
    \includegraphics[width=0.9\linewidth]{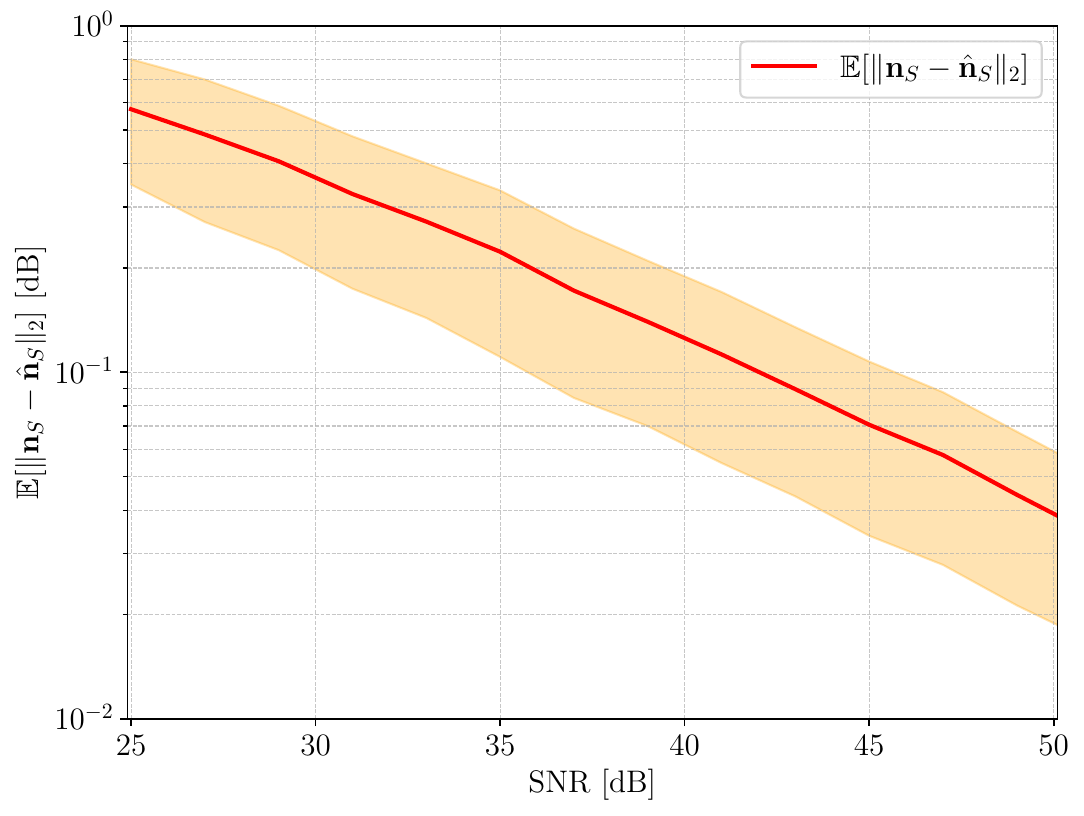}
    \caption{ Numerical evaluation of the expected bias norm $\E[\|\mathbf{n}_S - \hat{\mathbf{n}}_S \|_2]$ across varying SNR levels [\SI{}{\dB}]. The shaded area represents one standard deviation from the mean bias norm.}
    \label{fig:Bais_n_arrow}
\end{figure}

\begin{proposition}
\label{proposition:crlb bold_n}
Given the  estimator for the LED's normal vector $\hat{\mathbf{n}}_S$, as $\sigma^2 \rightarrow 0$, the expected bias norm is  $\E \left[\| \mathbf{n}_S-\hat{\mathbf{n}}_S \|_2 \right]$  asymptotically converges to 
\begin{equation}
    \label{eq:bias_norm}
 \E[ \|\mathbf{n}_S-\hat{\mathbf{n}}_S \|_2] = \frac{O(\sigma^2)}{c^2} \sqrt{ \mathbf{n}_S^T (\mat{G}\mat{G}^T)^{-2} \mathbf{n}_S  } .
\end{equation}
\end{proposition}

\begin{proof}
    Please see Appendix.~\ref{ap:4}.
\end{proof}

Fig.~\ref{fig:Bais_n_arrow} shows the simulation results for the expected bias norm $\E[\| \mathbf{n}_S - \hat{\mathbf{n}}_S \|_2 ]$, which demonstrates its asymptotic convergence as the noise variance $\sigma^2 \rightarrow 0$. In this simulation, the LED is mounted at the origin with a height of \SI{5.71}{\m} (same height as the factory environment in Section \ref{sec:factory_result}), and the LED tilt angles are set to $\Delta \theta = \SI{2.5}{\degree}$ and $\Delta \gamma = \SI{135}{\degree}$, 
with $c = 3$. The training datasets are uniformly sampled from a circle with a radius of \SI{2}{\m} around the LED. The LED transmitted optical power $P=1$. For each signal-to-Noise (SNR) level, 1000 simulation runs were performed. As shown in Fig.~\ref{fig:Bais_n_arrow}, the bias norm is inversely proportional to the SNR (i.e., proportional to $\sigma^2$).

\begin{figure}[!t]
    \centering
    \includegraphics[width=\linewidth]{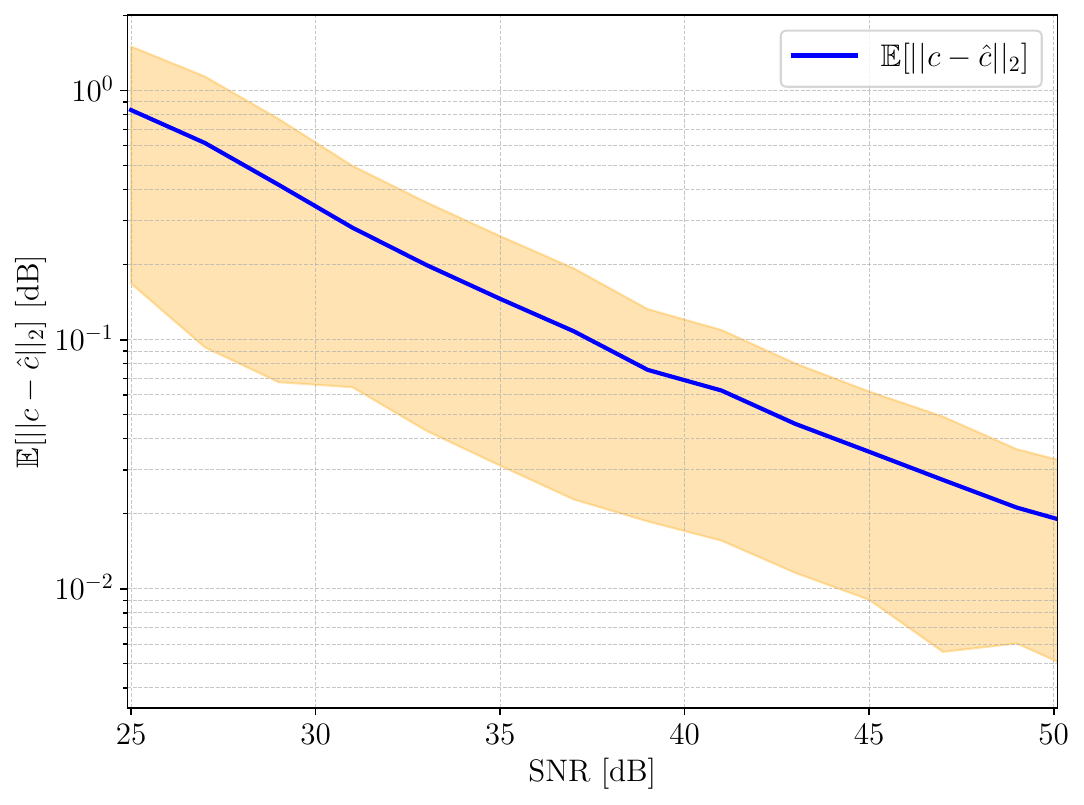}
    \caption{ Numerical evaluation of the expected bias norm $\E[|| {c - \hat{c}}||_2]$ across varying SNR levels [\SI{}{\dB}].}
    \label{fig:Bais_c_arrow}
\end{figure}

\begin{proposition}
\label{proposition:crlb c}
For the estimator of the scalar $\hat{c}$, as  $\sigma^2 \rightarrow 0$, the expected bias norm $\E[|| {c-\hat{c}} ||_2]$ asymptotically converges to
\begin{equation}
    \label{eq:bias_c}
\E[ || {c - \hat{c}} ||_2]  = O(\sigma^2).
\end{equation}
where $\vect{\nu} = [\nu_0,...,\nu_{n},...\nu_{N-1}]^T$  consists of elements $\nu_n$ that are i.i.d Gaussian unit distribution, and $\vect{w} = \sigma \vect{\nu}$.
\end{proposition}

\begin{proof}
    Please see Appendix.~\ref{ap:3}.
\end{proof}

Using the same configuration as in Fig.~\ref{fig:Bais_n_arrow},  The numerical results for the bias norm $\E[|| {c - \hat{c}} ||_2]$ are illustrated in Fig.~\ref{fig:Bais_c_arrow}. The bias norm is also inversely proportional to the SNR, that is, proportional to $\sigma^2$.

\subsection{Optimal Calibration Scheme}
The calibration accuracy depends on both the number of calibration points $N$, and their spatial distribution. Despite the precision gains from a higher $N$, exact coordinate measurement of the receiver is cost and time-consuming. Our strategy, therefore, focuses on optimizing the distribution of receiver positions for calibration at ground level. To achieve this, our goal is to minimize the sum \gls{mse} of the calibration parameters $\hat{\mathbf{c}}$, as formulated in   (\ref{eq:mini_summse}). We express this optimization goal as
\begin{equation}
\label{eq:optimaldataset}
   \min\limits _{\mat{G}} \sigma^2 \tr \left[ (\mat{GG}^T)^{-1} \right].
\end{equation}

\begin{theorem}
\label{theorem:optimal_calibration_dataset}
For $N \geq 3$, calibration points evenly distributed on a circle with radius $r^*=\sqrt{\frac{-3+\sqrt{13}}{2}} h \approx 0.55 h$ around the LED is optimal for LED calibration on the ground. The coordinates of these points can be represented as 
$$\vect{r}_{R,n} = [ r^* \cos \left(\frac{2 \pi}{N} n + \phi\right), r^* \sin \left(\frac{2 \pi}{N} n + \phi\right), 0]^T$$
where $n=0,...,N-1$ and $\phi$ is an arbitrary angle. Then the matrix $\mat{GG}^T$ will be simplified as
$$
\mathbf{G G}^T=\frac{ h^2}{\left(h^2+r^{*2}\right)^4} \operatorname{diag}\left[\frac{Nr^{*2}}{2}, \frac{Nr^{*2}}{2}, Nh^2\right].
$$

The minimum sum \gls{mse} is 
$$
\E\big[\norm{\vect{c} - \hat{\vect{c}}}^2\big] = \sigma^2  \frac{(h^2+ r^{*2})^4}{Nh^2}(\frac{4}{r^{*2}} + \frac{1}{h^2}) \approx \frac{ 40.94  \sigma^2 h^4 }{N}.
$$
\end{theorem}

\begin{proof}
    Please see Appendix.~\ref{ap:optimal_calibrationdata}.
\end{proof}

\begin{figure}[!t]
    \centering
    \tdplotsetmaincoords{70}{30}
    \begin{tikzpicture}[scale=2,tdplot_main_coords]
    \coordinate (O) at (0,0,0);
    \coordinate (LED) at (0,0,1.5);
   \coordinate (R0) at ({2*cos(0)},{2*sin(0)},0);
    \coordinate (R1) at ({2*cos(72)},{2*sin(72)},0);
    \coordinate (R2) at ({2*cos(144)},{2*sin(144)},0);
    \coordinate (R3) at ({2*cos(216)},{2*sin(216)},0);
    \coordinate (R4) at ({2*cos(288)},{2*sin(288)},0);

    \draw[fill=yellow] (LED) circle (0.1) node[above] {LED $\mathbf{r}_S = [0,0,h]^T$};

    \draw[thick,->] (-2,0,0) -- (2,0,0) node[anchor=north west]{$x$};
    \draw[thick,->] (0,-2,0) -- (0,2,0) node[anchor=south east]{$y$};
    \draw[thick,->] (O) -- (0,0,1.5) node[]{};
    
    \draw[thick,->,dashed] (O) -- ({2*cos(-135)},{2*sin(-135)},0) node[midway, sloped, above]{$r$};
    
    \tdplotdrawarc[tdplot_main_coords,gray]{(O)}{2}{0}{360}{}{}
    
    \node[label={[purple]above left:$\mathbf{r}_{R,0}$},circle,fill=red,inner sep=2pt] at (R0) {};
    \node[label={[purple]below left:$\mathbf{r}_{R,1}$},circle,fill=red,inner sep=2pt] at (R1) {};
    \node[label={[purple]below right:$\mathbf{r}_{R,2}$},circle,fill=red,inner sep=2pt] at (R2) {};
    \node[label={[purple]above right:$\mathbf{r}_{R,3}$},circle,fill=red,inner sep=2pt] at (R3) {};
    \node[label={[purple]above:$\mathbf{r}_{R,4}$},circle,fill=red,inner sep=2pt] at (R4) {};
    
    \end{tikzpicture}
\caption{Illustration of optimal calibration schemes with an LED at position $\mathbf{r}_S = [0,0,h]^T$ and five calibration points (red points) uniformly distributed on a circle of radius $r$ on the ground plane.}
\label{fig:cali_sche}
\end{figure}

\begin{figure}[!t]
    \centering
    \includesvg[width=1\linewidth]{Images/MSE_trace_0.001.svg}
    \caption{Variation of sum MSE of $\hat{\vect{c}}$ with changing radius $r$, demonstrating the adherence to the lower bound and optimal  at $r^*$.}
    \label{fig:optimal_calibration_result}
\end{figure}

Fig.~\ref{fig:cali_sche} depicts our calibration scheme. An LED is mounted on the ceiling with coordinates $\vect{r}_S=[0,0,h]^T$. The $5$ red points selected as calibration data are evenly distributed on a circle with radius $r$. In Fig.~\ref{fig:optimal_calibration_result},  as we vary  $r$ from $0.05h$ to $1.5h$, the sum MSE of $\hat{\vect{c}}$  first decreases, reaching its minimum around $0.55h$, and then increases. Notably, the sum MSE consistently aligns with the lower bound throughout this range. This observation confirms that our calibration scheme  always adheres to  the lower bound, and achieves the optimal performance when $r^*=\sqrt{\frac{-3+\sqrt{13}}{2}} h \approx 0.55 h$.

\section{Localization  Methodology}
\label{sec:LOCALIZATION}
Following calibration, the focus shifts to localization. For clarity, localization parameters will be denoted with a tilde notation. We also introduce the LED index $l$, where $l=0,1, \ldots, L-1$, with  $L$ representing the total number of LEDs. We employ $\tilde{s}_l$ to denote the \gls{rss} value at PD from the LED $l$ during localization, while  $\tilde{\vect{r}}_R$ and $\vect{r}_{S,l}$ denote the coordinates of the PD and the LED $l$, respectively. Furthermore, we define $\tilde{d}_l$ as the distance, expressed as $\tilde{d}_l = \norm{\tilde{\vect{r}}_R - \vect{r}_{S,l} }$, and introduce $\tilde{\vect{d}}_l=\tilde{\vect{r}}_R - \vect{r}_{S,l}$. 

We still assume that the receiver PD has no tilt and is located on the ground and that LED $l$ is located at a height $h_l$. From the calibration process, we obtain the estimated value $\hat{\vect{n}}_{S,l}$ and $\hat{c}_l$ for LED $l$. The RSS $\tilde{s}_l$ can be expressed as 
\begin{equation}
\label{eq:localization}
\tilde{s}_l=\frac{\hat{c}_l {h}_l}{\tilde{d}_l^4} \hat{\mathbf{n}}_{S, l}^T \tilde{\vect{d}}_l
 +e_l+\tilde{w}_l
\end{equation}
where $\tilde{w}_l \sim \mathcal{N}\left(0, \sigma_l^2\right)$ is the localization noise. Assuming the same variance $\sigma_l^2$ for both calibration noise and localization noise ($\tilde{w}_l$ and $w_l$) is practical due to the consistent noise characteristics of the same LED. In practical application, the variance $\sigma_l^2$ can be obtained  directly from the sensor's datasheet or estimated using   (\ref{eq:noise_estimate}) following calibration. The variable $ e_l$ is the error arising from calibration noise, defined as 
$$
e_l =\frac{c_l {h}_l}{\tilde{d}_l^4} \mathbf{n}_{S, l}^T \tilde{\vect{d}}_l-\frac{\hat{c}_l {h}_l}{\tilde{d}_l^4} \hat{\mathbf{n}}_{S, l}^T \tilde{\vect{d}}_l. 
$$

Employing   (\ref{eq:answer}) and   (\ref{eq:newrangemodel}), $e_l$ can be further simplified
$$
\begin{aligned}
e_l &= \frac{{h}_l}{\tilde{d}_l^4}(c_l \mathbf{n}_{S, l}-(\mathbf{G}_l \mathbf{G}_l^T)^{-1} \mathbf{G}_l \mathbf{s}_l)^T  \tilde{\vect{d}}_l  \\
&=-\frac{{h}_l}{\tilde{d}_l^4} \tilde{\vect{d}}_l^T (\mathbf{G}_l \mathbf{G}_l^T)^{-1} \mathbf{G}_l \mathbf{w}_l
\end{aligned}
$$
where $\vect{w}_l$ is a vector consisting of $N$ samples of $w_l$, representing the calibration noise specific to LED $l$. The $\mat{G}_l$ is structured similarly to $\mat{G}$ but includes the subscript $l$ to identify the specific LED during calibration.  Given the i.i.d zero mean nature of ${w}_l$, $e_l$ has zero mean and and variance given by  
\begin{equation}
\label{eq:E_e_square}
    \E\left[e_l^2\right]=\sigma_l^2 \frac{{h}_l^2}{\tilde{d}_l^8} \tilde{\vect{d}}_l^T \left(\mathbf{G}_l \mathbf{G}_l^T\right)^{-1} \tilde{\vect{d}}_l.
\end{equation}

Thus, equation (\ref{eq:localization}) can be rewritten as $\tilde{s}_l=\mu_l+n_l,$
where 
$$\mu_l=\frac{\hat{c}_l {h}_l}{\tilde{d}_l^4} \hat{\mathbf{n}}_{S, l}^T \tilde{\vect{d}}_l$$
and  the total noise $n_l$ is
$n_l=e_l + \tilde{w}_l$ with $\E [n_l]=0$, with variance
$$
\begin{aligned}
\E\left[n^2_l\right] &= \E\left[e_l^2\right]+\mathbb{E}\left[\tilde{w}_l^2\right] \\
&=\sigma_l^2\left(\frac{h_l^2}{\tilde{d}_l^8} \tilde{\vect{d}}_l^T (\mathbf{G}_l \mathbf{G}_l^T)^{-1} \tilde{\vect{d}}_l+1\right).
\end{aligned}
$$
As can be seen, the noise variance $\E [n_l^2]$ is LED dependent even if the noise variance $\sigma_l^2$  are equal for each LED, and the localization \gls{ml} problem becomes equivalent to a  weighted \gls{ls} problem
\begin{equation}
\label{eq:wml}
\min _{\tilde{\mathbf{r}}_R} \sum_{l=0}^{L-1} \frac{1}{\E \left[n_l^2\right]}\left(\tilde{s}_l-\frac{\hat{c}_l h_l}{\tilde{d}_l^4} \hat{\mathbf{n}}_{S, l}^T \tilde{\vect{d}}_l \right)^2 .
\end{equation}

Incorporating this understanding with our calibration process, we introduce a method utilizing   weighted \gls{ls}, as detailed in Algorithm~\ref{algo_LS}.

\begin{algorithm}[tb]
\SetKwData{Left}{left}\SetKwData{This}{this}\SetKwData{Up}{up}
\SetKwFunction{Union}{Union}\SetKwFunction{FindCompress}{FindCompress}
\SetKwInOut{Input}{input}\SetKwInOut{Output}{output}
\Input{

Calibration Measurements for LED $l$: \{$\vect{r}_{R,n},s_n,\vect{r}_{S,l}$\},

RSS during localization: \{$\tilde{s}_l$\},

where $n=0,...N-1$, and $l=0,..,L-1$.}
\Output{Localization coordinates $\tilde{\vect{r}}_R$}
\BlankLine
// Stage 1. Calibrate tilt and gain and estimate noise variance for each LED.\\
\For{\text{LED} $l$ \KwTo $L-1$}{
    Use   (\ref{eq:estimate_n_c}) to calibrate $\hat{\vect{n}}_{S,l}$ and $\hat{c}_l$. \\
    Use   (\ref{eq:noise_estimate}) to estimate $\hat{\sigma}_l^2$.
}
\BlankLine
//Stage 2. Perform localization. \\
Use   (\ref{eq:wml}) to compute $\tilde{\vect{r}}_R$.

\caption{  weighted \gls{ls}}
\label{algo_LS}
\end{algorithm}

Given the uncorrelated nature of observations from each individual LED $l$, the overall \gls{fim} can be written as the sum of the matrices related to each LED
\begin{equation}
\label{eq:fim_ml}
    \mathbf{I}(\tilde{\mathbf{r}}_R) = \sum _{l=0}^{L-1}  \mathbf{I}_l(\tilde{\mathbf{r}}_{R}).
\end{equation}

Utilizing the general result of (3.31,~\cite{kay1993fundamentals}), the \gls{fim} associated with $\tilde{\vect{r}}_R$ for LED $l$ is then defined as 

\begin{equation}
\label{eq:localization CRLB}
\begin{aligned}
\mathbf{I}_l(\tilde{\mathbf{r}}_{R}) & =\frac{1}{\E \left[n^2_l\right]}\frac{\mathrm{d} \mu_l}{\mathrm{d} \tilde{\mathbf{r}}_R} \frac{\mathrm{d} \mu_l}{\mathrm{d} \tilde{\mathbf{r}}_R^T} + \frac{1}{2} \frac{1}{(\E\left[n^2_l\right])^2} \frac{\mathrm{d} \E \left[n_l^2\right]}{\mathrm{d} \tilde{\mathbf{r}}_R} \frac{\mathrm{d} \E\left[n_l^2\right]}{\mathrm{d} \tilde{\mathbf{r}}_R^T}.
\end{aligned}
\end{equation}

Introducing $\mathbf{\Omega}_{l}=\mathbf{G}_l \mathbf{G}_l^T \in \mathbb{R}^{3\times 3}$, the following relations in   (\ref{eq:localization CRLB}) are obtained
$$
\begin{aligned}
\frac{\mathrm{d} \mu_l}{\mathrm{d} \tilde{\mathbf{r}}_R} &= -4\frac{\hat{c_l}{h}_l}{\tilde{d}_l^{6}} \tilde{\mathbf{d}}_l (\hat{\mathbf{n}}_{S, l}^T \tilde{\mathbf{d}}_l) + \frac{\hat{c_l} {h}_l}{\tilde{d}_l^{4}} \hat{\mathbf{n}}_{S, l} \\
\end{aligned}
$$
and
$$
\begin{aligned}
   \frac{\mathrm{d} \E\left[n_l^2\right]}{\mathrm{d} \tilde{\mathbf{r}}_R} &= \sigma_l^2 \frac{-8 {h}_l^2}{\tilde{d_l}^{10}} \tilde{\mathbf{d}}_l (\tilde{\mathbf{d}}_l^T \mathbf{\Omega}_{l}^{-1} \tilde{\mathbf{d}}_l )
  + \sigma_l^2 \frac{2 {h}_l^2}{\tilde{d_l}^{8}} \mathbf{\Omega}_{l}^{-1} \tilde{\mathbf{d}}_l. \\
\end{aligned}
$$

Finally, the \gls{crlb} is obtained as  $ \mathbf{C}({\tilde{\mathbf{r}}}_R) = \mathbf{I}^{-1}(\tilde{\mathbf{r}}_R)$. In the matrix $\mat{C}(\tilde{\vect{r}}_R)$ , the diagonal elements $\mat{C}(\tilde{\vect{r}}_R)$ represent the variance lower bound of the  $x,y$ and $z$ axes, respectively. Given that the object is always on the ground, we focus solely on the variance of $x$ and $y$ axis. The CRLB for the   weighted \gls{ls} estimator, is calculated as 
\begin{equation}
\label{eq:crlb_ _ml}
    \operatorname{CRLB} = \sqrt{ [\mat{C}(\tilde{\vect{r}}_R)]_{0,0} + [\mat{C}(\tilde{\vect{r}}_R)]_{1,1} }
\end{equation}

\section{Simulation and Experimental Results}
\label{sec:SIMULATION AND EXPERIMENTAL RESULTS}

\subsection{Simulation Localization Results}

\begin{table}[!t]
    \centering
    \caption{Simulation experimental set-up for various tilt angles and gain analysis}
    \label{tab:simulation_setup_tilt_gain}
    \begin{tabular}{lc}
    \hline
    \textbf{ Parameters} & \textbf{ Values} \\ \hline
    Room Dimensions                & \SI{8}{\meter} $\times$ \SI{8}{\meter}                                            \\
    Height $h$                     & \SI{5.71}{\meter}                                                   \\
    Number of LEDs                 & 4                                                    \\
    LEDs position [\SI{}{\meter}] & \begin{tabular}[c]{@{}c@{}}{[}-2, 6{]}, {[}2, 6{]}\\ {[}-2, 2{]}, {[}2, 2{]}\end{tabular} \\
    \textbf{LEDs tilt polar angle} ${\Delta \theta_l}$ & $\mathcal{N}(\SI{0}{\degree}, (\SI{2}{\degree})^2 )$ \\
    \textbf{LEDs tilt azimuthal angle} $\Delta \gamma_l$     & Random $\SI{0}{\degree} \sim \SI{360}{\degree} $\\
    LED radiation pattern &    Lambertian with $m=1$  \\
    \textbf{Gains} $c_l$ & Random $1 \sim 3$ \\
    Noise standard deviation $\sigma_l$                   &    $0.0001$ (uniform for all LEDs)\\
    \hline
    \end{tabular}
    \end{table}

    Tab.~\ref{tab:simulation_setup_tilt_gain} summarizes the parameters used in  our calibration and localization simulation. The square room with a size of  \SI{64}{\m^2}  has a height of  \SI{5.71}{\meter} (same height as the factory environment in Section \ref{sec:factory_result}), equipped with 4 LEDs. These LEDs are slightly tilted, with their polar angles $\Delta \theta_l$ following a Gaussian distribution with zero mean and standard deviation of \SI{2}{\degree},  while their  azimuthal angles $\Delta \gamma_l$ are uniformly distributed between $\SI{0}{\degree} \sim \SI{360}{\degree} $. The gain $c_l$ are also  drawn from a uniform distribution within the range of $1 \sim 3$. This configuration includes varying tilt angles, a Lambertian radiation pattern, and different gains for each LED, is further wrapped up by a Gaussian noise model.

    \begin{figure}[!t]
        \centering
        \includegraphics[width=\linewidth]{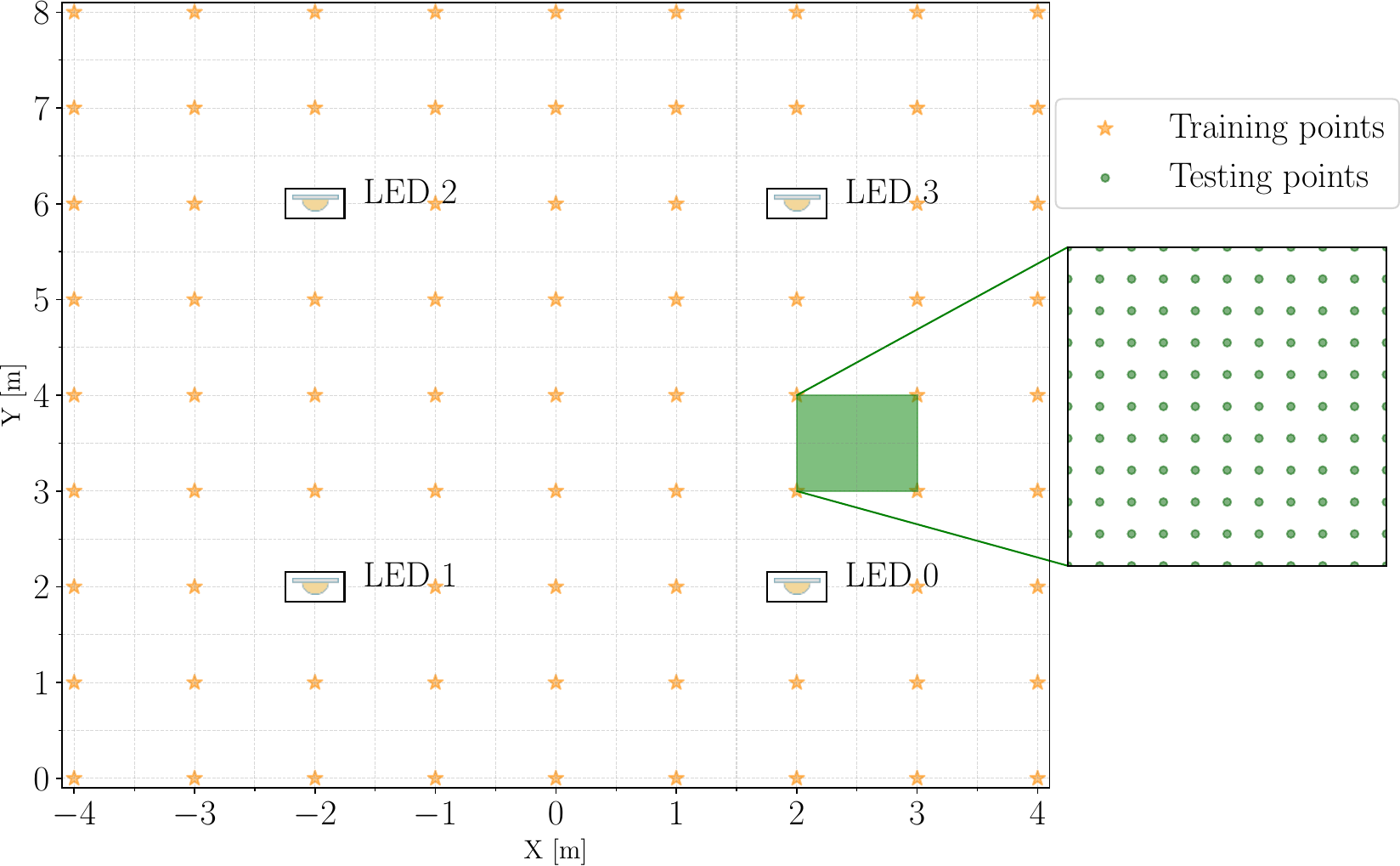}
        \caption{ Distribution of training and testing datasets. The training datasets are yellow stars, uniformly sampled with various sizes \(N\). The testing datasets, represented by green points, are uniformly separated at intervals of \SI{10}{\cm}, totaling $6561$ points.}
        \label{fig:simulation_angles}
    \end{figure}

    Fig.~\ref{fig:simulation_angles} illustrates the training datasets and testing datasets. The training sets are uniformly sampled within the ranges \([-4,4]\) and \([0,8]\) for the X and Y axes on the map, regardless of  the  training set size. We repeated the experiment 50 times for each training set size $N$ to  to evaluate the effects of tilt angels $\vect{n}_S$ and gain $c$. The training set sizes  are $N = \{ 9,16,25,36,64\}$. The testing data size keeps the same number of $6561$ points.

\begin{figure}[tb]
    \centering
    \includegraphics[width=\linewidth]{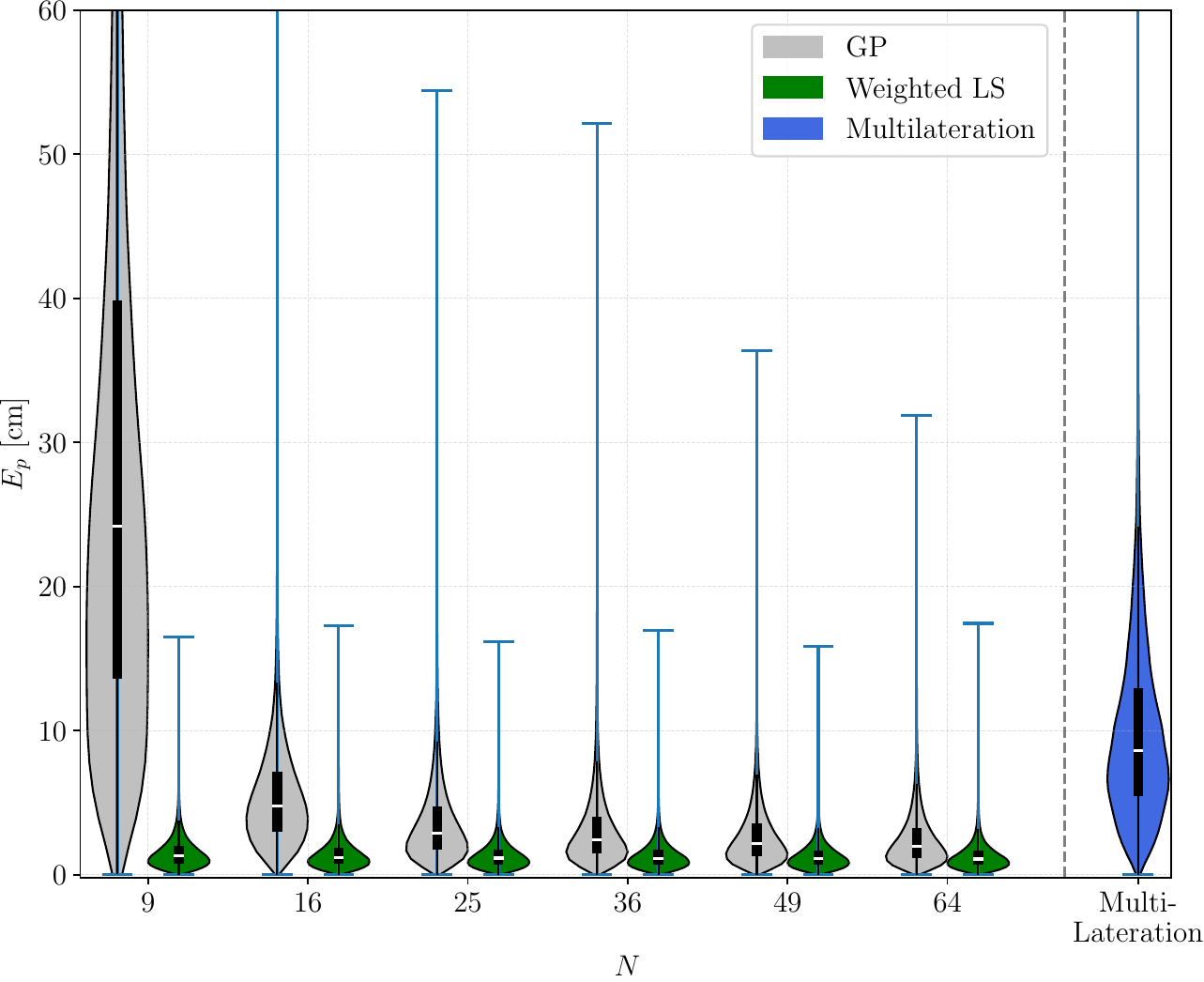}
    \caption{Simulations Results: Comparison between GPs, Weighted LS and Multilateration method for various simulated training datasets.}
    \label{fig:benchmark_boxplot}
\end{figure}

The simulation results are shown in Fig.~\ref{fig:benchmark_boxplot}, the  horizontal axis quantifies the localization error in terms of Euclidean distance, defined as
$$
E_p = \sqrt{(\hat{x} - x)^2 +(\hat{y}-y)^2}
$$
where $\hat{x},\hat{y}$ are the estimated coordinates, and $x,y$ are the ground truth coordinates. The multilateration model \cite{zhuang2018survey} has the assumption of no tilt angles for the LEDs and PD, formalized by  $\cos(\phi_{S,n}) = \cos(\phi_{R,n}) = h/d_n$ for each LED. This specific assumption establishes the direct link between RSS values and distances as  $d_n^4 \propto 1/s_n$. Subsequently, $L$ over-determined equations, corresponding to the number of LEDs, are employed to estimate the 2D coordinates of the PD. However, it is  shown in Fig.~\ref{fig:benchmark_boxplot} that these assumptions introduce errors,  with the P50 error being \SI{8.59}{\cm} and the P99 error being \SI{31.05}{\cm} in Tab.~\ref{tab:simulation_result}.

\renewcommand{\arraystretch}{1.2} 
\begin{table}[tb]
\caption{Simulation results: 50 and 99 percentile of the error (unit cm)}
\label{tab:simulation_result}
\resizebox{0.495\textwidth}{!}{%
\begin{tabular}{ccccccc}
    \hline
    Conf. & P50 & P99 & P50 & P99 & P50 & P99 \\ \hline
     & \multicolumn{2}{c}{$N=9$} & \multicolumn{2}{c}{$N=16$} & \multicolumn{2}{c}{$N=25$} \\
    \gls{gp} $E_p$ & 24.18 & 128.34 & 4.80 & 19.15 & 2.91 & 14.68 \\
    weighted \gls{ls} $E_p$ & 1.29 & 5.60 & 1.21 & 5.21 & 1.16 & 4.99 \\
     & \multicolumn{2}{c}{$N=36$} & \multicolumn{2}{c}{$N=49$} & \multicolumn{2}{c}{$N=64$} \\
    \gls{gp} $E_p$ & 2.45 & 12.86 & 2.17 & 10.67 & 1.96 & 9.60 \\
    weighted \gls{ls} $E_p$ & 1.15 & 4.95 & 1.12 & 4.77 & 1.09 & 4.73 \\ \hline
    Conf. & \multicolumn{3}{c}{P50}   &  \multicolumn{3}{c}{P99}  \\ 
    Multilateration $E_p$ & \multicolumn{3}{c}{8.59}   &  \multicolumn{3}{c}{31.05} \\ \hline
    \end{tabular}
}
\end{table}
\renewcommand{\arraystretch}{1} 

As shown in the Tab.~\ref{tab:simulation_result}, the weighted \gls{ls} method always outperforms the \gls{gp} method across all training size $N$, especially in reducing P99 errors. For example, with a training set size of $N=9$, the \gls{gp} method produces significant errors with P50 and P99 values of \SI{24.18}{\cm} and \SI{128.34}{cm}, respectively. In contrast, the weighted \gls{ls} method reduces these errors to \SI{1.29}{\cm} and \SI{5.6}{\cm}, respectively. When the training set size gets bigger, both methods show improvement, but weighted \gls{ls} consistently demonstrates superior performance, particularly in minimizing P99 errors, which is important for applications that require high reliability.

\subsection{Experimental Localization Results}
\subsubsection{Result on Platform Environment}

\begin{figure}[tb]
    \centering
    \includesvg[width=0.8\linewidth]{Images/experimental_setup.svg}
    \caption{Platform setup featuring four LEDs mounted at a height of \SI{1.284}{\meter} within a \SI{3}{\meter} $\times$ \SI{3}{\meter} environment.}
    \label{fig:real_experimental_setup}
\end{figure}

To evaluate the performance of the   weighted \gls{ls} algorithm, we designed and implemented an experimental setup, illustrated in Fig.~\ref{fig:real_experimental_setup}. The setup includes four LEDs positioned at a height of \SI{1.284}{\meter}, with a PD capable of moving within a \SI{3}{\meter} $\times$ \SI{3}{\meter} plane. Tab.~\ref{tab:setup} below summarizes the key parameters of our experimental setup. We recorded measurements at $158$ distinct points, each with precisely determined coordinates and RSS value\footnote{\url{http://dx.doi.org/10.21227/f28n-6292}}.  Reflecting on the research by \cite{knudde2020data,raes2020experimental}, which showed \gls{gp} outperforming \gls{mlp} and multilateration in RSS-based \gls{vlp}, our study will not include \gls{mlp} as  a benchmark, focusing on the more effective \gls{gp} approach.

\begin{table}[tb]
\centering
\caption{Platform setup parameters}
\label{tab:setup}
\begin{tabular}{cc}
\hline
\textbf{Parameters}                  & \textbf{Values}                                                    \\ \hline
Height $h$                           & \SI{1.284}{\meter}                                                     \\
Number of LEDs                           & 4\\
Room dimensions                             & \SI{3}{\meter} $\times$ \SI{3}{\meter} \\
LEDs position [\SI{}{\meter}]                     & \begin{tabular}[c]{@{}c@{}}{[}$1.059$,~$2.470${]},~{[}$2.428$,~$2.552${]}\\ {[}$1.031$,~$0.630${]},~{[}$2.402$,~$0.582${]}\end{tabular} \\ 
Number of measurement points      &  $158$ \\

\hline \\
\end{tabular}
\end{table}

In our experiments, calibrated parameters $\hat{\mathbf{n}}_{S,l}$ and $\hat{c}_l$ were employed in the   weighted \gls{ls}  localization algorithm.  It should be noted that in this dataset, the majority of measurement points were centrally located on the map, while LEDs were positioned at each corner. This layout diverges from the optimal calibration scheme, which recommends positioning measurement points near the LEDs. Owing to this mismatch, direct validation of the optimal calibration scheme with this specific dataset was not viable. Therefore, we implemented a calibration approach involving uniform random sampling from the dataset. This method enabled us to evaluate the impact of LED tilt on positioning accuracy and to assess the performance of the   weighted  \gls{ls} algorithm under these conditions.

Fig.~\ref{fig:Platform_results_violine} offers a detailed comparative analysis of three localization techniques: multilateration, \gls{gp}, and   weighted \gls{ls}, with the horizontal axis denoting the size of the training set employed for calibration or model training. 

The black rectangle plots outline the interquartile range (P25 to P75), and the P50 error is indicated by the white line within each rectangle.  The weighted \gls{ls} outperforms \gls{gp} at smaller training set sizes, with \gls{gp} reaching a similar median performance to   weighted \gls{ls} at a training set size of $25$. Despite improvements with larger training sets, \gls{gp} consistently shows higher outlier errors than   weighted \gls{ls}, a reflection of machine learning algorithms' challenges with data-sparse regions. Multilateration, based on the assumption of non-tilted LEDs, consistently yields the least accurate outcomes, underscoring the significant impact of LED tilt on localization accuracy and the necessity of accounting for it in calibration.

\begin{figure}
    \centering
    \includesvg[width=1\linewidth]{Images/Measurement_result.svg}
    \caption{Platform results: comparison between \gls{gp}, multilateration and   weighted \gls{ls} methods for the different size of training set. The horizontal axis denotes the training size $N$. The vertical axis represent the Euclidean error. }
    \label{fig:Platform_results_violine}
\end{figure}

\renewcommand{\arraystretch}{1.2} 
\begin{table}[tb]
\caption{Platform results: 50 and 99 percentile of the error (unit cm) and the improvement rate of using  weighted  \gls{ls} over \gls{gp}}
\label{tab:platform_result}
\resizebox{0.495\textwidth}{!}{%
\begin{tabular}{ccccccc}
\hline
Conf. & P50 & P99 & P50 & P99 & P50 & P99 \\ \hline
 & \multicolumn{2}{c}{$N=9$} & \multicolumn{2}{c}{$N=16$} & \multicolumn{2}{c}{$N=25$} \\
\gls{gp} $E_p$ & 7.95 & 42.80 & 4.88 & 29.91 & 4.29 & 19.99 \\
 weighted  \gls{ls} $E_p$ & 3.24 & 11.08 & 3.16 & 11.64 & 3.60 & 13.67 \\
Improvement rate [\%] & \textbf{58.7} & \textbf{74.1} & 35.3 & 61.1 & 16.1 & 31.6 \\
 & \multicolumn{2}{c}{$N=36$} & \multicolumn{2}{c}{$N=49$} & \multicolumn{2}{c}{$N=64$} \\
\gls{gp} $E_p$ & 3.28 & 11.72 & 3.20 & 13.99 & 2.88 & 10.00 \\
 weighted  \gls{ls} $E_p$ & 3.11 & 11.44 & 2.87 & 9.54 & 2.64 & 7.56 \\
Improvement rate [\%] & 5.2 & 2.4 & 10.3 & 31.8 & 8.3 & 24.4 \\ \hline
\end{tabular}
}
\end{table}
\renewcommand{\arraystretch}{1} 

Tab.~\ref{tab:platform_result} quantifies the localization error at P50 and P99, respectively. The improvement rate is defined by
$$
\text{Improvement rate [\%]} =\frac{\text{\gls{gp} } E_p - \text{Weighted   \gls{ls} } E_p}{ \text{\gls{gp} } E_p  }.
$$

With a minimal training set size of $N=9$, our weighted  \gls{ls} method already reaches its best performance, achieving significant gains of \SI{74}{\%} at the P99 and  \SI{59}{\%} at the P50 over \gls{gp}. This efficiency highlights the method's advantage of requiring only a small number of data points for optimal calibration. Even as the training set size increases, with the most modest gain observed at \SI{2.4}{\%} at the P99 for $N=36$, weighted \gls{ls} continues to outperform \gls{gp} across the board, proving its consistent effectiveness in enhancing localization accuracy with easy calibration. In addition, the P50 and P99 for multilateration are \SI{7.4}{\cm} and \SI{25.7}{\cm}, respectively.

\begin{figure}[!t]
    \centering
    \includesvg[width=0.9\linewidth]{Images/CRLB_Wegithed_LS_Result.svg}
    \caption{Platform results: The red box represents the \gls{crlb} of  weighted  \gls{ls}. The green box displays the  weighted  \gls{ls} as same as Fig.~\ref{fig:Platform_results_violine}. The horizontal axis represents the training size $N$. The vertical axis represents the Euclidean error or \gls{crlb}.}
    \label{fig:Estimate_Lowerbound}
\end{figure}

Using the same calibration measurements as Fig.~\ref{fig:Platform_results_violine}, we calculated the \gls{crlb} for each point in the test dataset via   (\ref{eq:crlb_ _ml}), by integrating estimated tilt and gain. In Fig.~\ref{fig:Estimate_Lowerbound}, the white line within each box denotes the P50 metric. When comparing various training dataset sizes, the theoretical variance of the weighted \gls{ls} method appears to stabilize, as indicated by the consistent height of the red boxes across different training sizes. This observation also accounts for the persistence of the P50 metric despite increasing the training set size beyond 36, suggesting that augmenting the training set does not substantially impact the median error.

\subsubsection{Result in Factory Environment}
\label{sec:factory_result}

\begin{figure}[tp]
    \centering
    \includegraphics[width=0.85\linewidth]{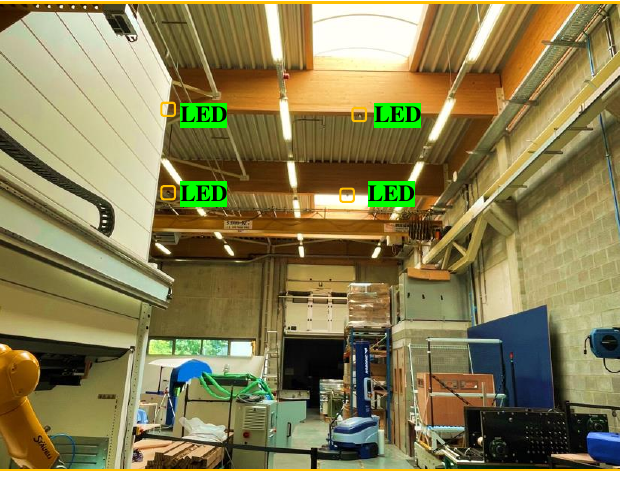}
    \caption{Factory setup featuring four LEDs mounted at a height of \SI{5.71}{\meter} within a \SI{7}{\meter} $\times$ \SI{7}{\meter} environment. Further details on the setup can be found in \cite{raes2021cellular}.}
    \label{fig:factory_environment}
\end{figure}

To check how well the  weighted \gls{ls} method  works in a real-world environment, we utilized a dataset \cite{raes2021cellular} from our group collected in a factory setting, as depicted in Fig.~\ref{fig:factory_environment}. The LEDs were mounted at a height of \SI{5.71}{\meter}, and the measurements were taken along a trajectory within a \SI{7}{\meter} $\times$ \SI{7}{\meter} area on the ground. The raw data were subsampled to ensure a minimum interval of \SI{20}{\centi\meter} between each measurement point. As detailed in Tab.~\ref{tab:factory_setups}, the total number of measurement points is 292. To evaluate the performance of the localization methods, we randomly selected the training dataset with different sizes $N = \{6, 9, 16, 25, 36, 49\}$, and repeated 10 times to obtain statistically relevant results.

\begin{table}[tb]
    \centering
    \caption{Factory setup parameters}
    \label{tab:factory_setups}
    \begin{tabular}{cc}
    \hline
    \textbf{Parameters}                  & \textbf{Values}                                                    \\ \hline
    Height $h$                           & \SI{5.71}{\meter}                                                     \\
    Number of LEDs                           & 4\\
    Room dimensions                             & \SI{7}{\meter} $\times$ \SI{7}{\meter} \\
    LEDs position [\SI{}{\meter}]                     & \begin{tabular}[c]{@{}c@{}}{[}$4$,~$6${]},~{[}$8$,~$6${]}\\ {[}$4$,~$0${]},~{[}$8$,~$0${]}\end{tabular} \\ 
    Number of measurement points     &  $292$ \\
    Minimal interval between each points & \SI{20}{\cm} \\
    \hline \\
    \end{tabular}
    \end{table}

    Fig.~\ref{fig:factory_results} shows the violin plots that compare the accuracy of the \gls{gp}, multilateration, and weighted \gls{ls} methods with different training dataset sizes. The weighted \gls{ls} method always  surpasses the GP method for the P99 error, especially when the training dataset is small. As the  training size increases, the accuracy of the GP method improves, but its maximum $E_p$  remains approximately twice as large as that of the weighted LS method at training size $N=49$. 
    This discrepancy generated because GP inherently perform interpolation in the input space (i.e, 4-dimensional RSS values). If the training data does not sufficiently cover this entire 4-dimensional space, it becomes challenging for the GP method to minimize the maximum error.

\begin{figure}[tb]
    \centering
    \includegraphics[width=\linewidth]{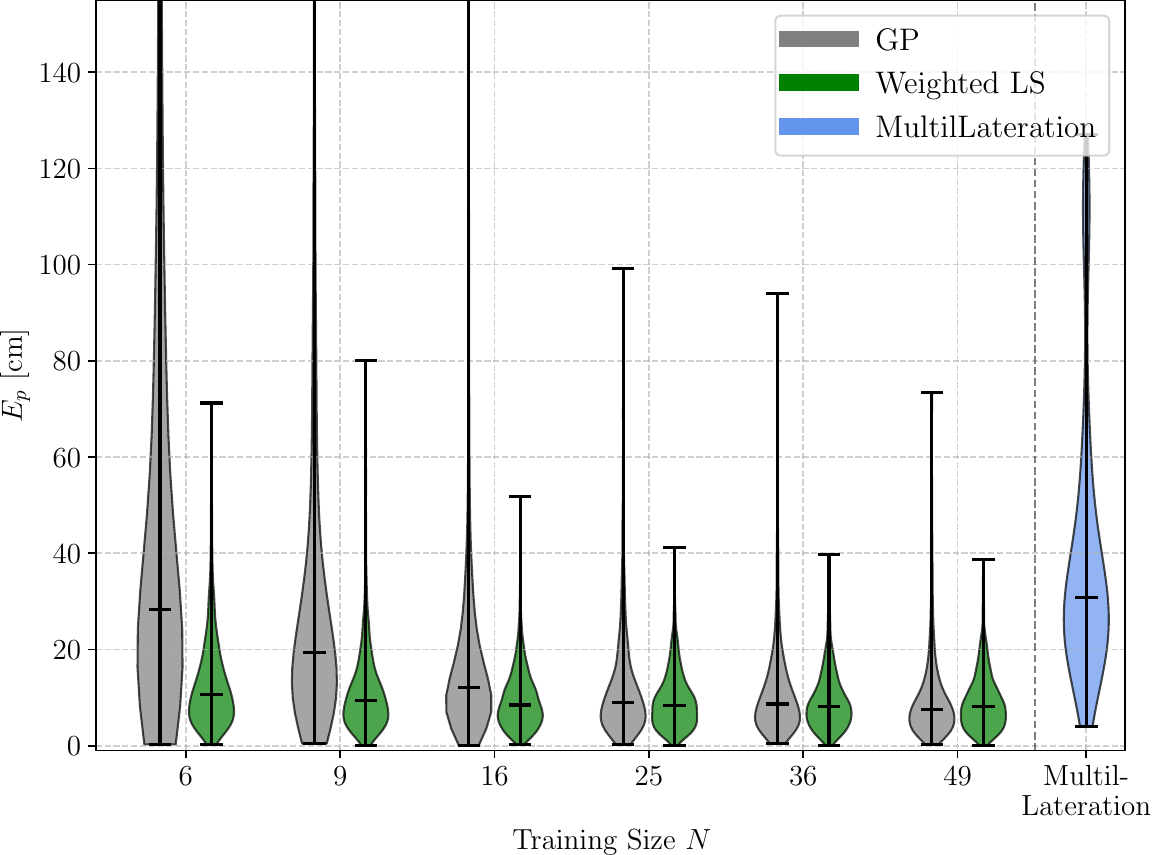}
    \caption{ Factory results: comparison between \gls{gp}, multilateration and   weighted \gls{ls} methods for the different training size $N$.}
    \label{fig:factory_results}
\end{figure}

\renewcommand{\arraystretch}{1.2} 
\begin{table}[!t]
\caption{Factory results: 50 and 99 percentile of the error (unit cm)}
\label{tab:factory_result}
\resizebox{0.495\textwidth}{!}{%
\begin{tabular}{ccccccc}
\hline
Conf. & P50 & P99 & P50 & P99 & P50 & P99 \\ \hline
 & \multicolumn{2}{c}{$N=6$} & \multicolumn{2}{c}{$N=9$} & \multicolumn{2}{c}{$N=16$} \\
GP $E_p$ & 28.40 & 513.73 & 19.39 & 239.58 & 12.07 & 112.12 \\
Physical Model $E_p$ & 10.64 & 47.47 & 9.36 & 43.85 & 8.46 & 28.46 \\
Improvement rate [\%] & \textbf{62.54} & \textbf{90.75} & 51.71 & 81.70 & 29.86 & 74.63 \\
 & \multicolumn{2}{c}{$N=25$} & \multicolumn{2}{c}{$N=36$} & \multicolumn{2}{c}{$N=49$} \\
GP $E_p$ & 9.02 & 55.70 & 8.67 & 58.67 & \textbf{7.95} & 43.43 \\
Physical Model $E_p$ & 8.39 & 28.22 & 8.21 & 29.27 & \textbf{8.04} & 29.73 \\
Improvement rate [\%] & 6.99 & 49.35 & 5.31 & 50.11 & \textbf{-1.1} & 31.55 \\ \hline
\end{tabular}
}
\end{table}
\renewcommand{\arraystretch}{1} 

Tab.~\ref{tab:factory_result} gives more details, showing that with just six training points, the weighted LS improves P99 by 90.75\% compared to GP. Even with larger training sets, such as $N=36$, the weighted LS method still  outperform GP at P99, with a 50.11\% improvement. At $N=49$, while GP slightly surpasses weighted LS at P50, its maximum error remains higher. This shows GP’s weakness in controlling maximum errors when training data is insufficient. Despite GP can theoretically achieve higher accuracy with more data, it lacks the interpretability of the weighted LS method, where each parameter has a clear physical meaning. Even with minimal data, weighted LS gives reliable results, achieving a P50 of \SI{10}{\cm} with just six points, making it a robust choice for scenarios with limited data or where model transparency is crucial.

However, we noticed that the accuracy of the results in the factory environment (Tab.~\ref{tab:factory_result}) is lower compared to those obtained on the simulation (Tab.~\ref{tab:simulation_result}). One possible reason (other reasons such as non Gaussian noise, different noise variance, multipath effect) is that in the factory, data was collected with a PD mounted on a moving vehicle, which may have caused slight shaking and made the PD hard to stay perfectly horizontal. This movement affects both the Weighted LS method and the GP algorithm. But the results still meet the requirements for most practical applications, even though the errors are bigger than  in the controlled platform experiments.

\section{Conclusion}
\label{sec:CONCLUSION}
This paper presented a comprehensive study on the calibration and localization in \gls{vlp} systems. The research began with a detailed explanation of the \gls{vlp} system and the development of a calibration model. We introduced a novel calibration framework that leverages an optimal calibration scheme to minimize \gls{mse}. Our development of a low-complexity, closed-form estimator for LED tilt and gain achieves the \gls{crlb} for \gls{mse}. 
While physical methods like lasers can measure LED tilt angles accurately, aligning these measurements with the global coordinate system needs extra devices, which adds complexity. Our approach simplifies this by directly estimating tilt angles in the global coordinate system, making it more practical for real-world applications.

In the localization phase, the  weighted  \gls{ls} algorithm was employed, which proved to be superior to \gls{gp} and multilateration. This underscores the   weighted \gls{ls} algorithm's capability to ensure dependable localization accuracy, achieving up to \SI{3.2}{\cm} at the P50 and \SI{11}{\cm} at the P99, even in environments with limited data.

Overall, this research enriches our understanding of \gls{vlp} systems, particularly in calibration and localization aspects. It underscores the pivotal role of factors like LED tilt and model-based interpretability, paving the way for future enhancements in \gls{vlp} system performance and the potential development of more effective physics-inspired machine learning approaches. 
An important perspective of our work is to investigate scenarios where the PD does not remain horizontal during movement. Interesting mitigation techniques may include integrating an inertial measurement unit (IMU) on the PD to monitor its tilt in real time.

\section*{Acknowledgment}
The authors thank Willem Raes and Jorik De Bruycker \cite{knudde2020data,raes2020experimental} for their crucial measurement contributions to this research. Artifical intelligence (AI) was used for finding (code) templates, debugging, editing and grammar enhancement.

\appendix
\section{APPENDIX}

\subsection{Proof of Proposition \ref{proposition:crlb bold_c}}
\label{ap:2}
Combing    (\ref{eq:answer}) and   (\ref{eq:newrangemodel}), The expectation error of  $\hat{\mathbf{c}}$  is found to be
$$
\begin{aligned}
\E \left[\mathbf{c}-\hat{\mathbf{c}}\right] & = \E \left[\mathbf{c}-(\mathbf{G G}^T)^{-1} \mathbf{G} \mathbf{s}\right]  \\
& = \E \left[ \mathbf{c}-(\mathbf{G G}^T)^{-1} \mathbf{G}(\mathbf{G}^T \mathbf{c}+\mathbf{w}) \right] \\
& = \E \left[-(\mathbf{G G}^T)^{-1} \mathbf{G} \mathbf{w} \right] \\
& = \mathbf{0}
\end{aligned}
$$
which denotes it's an unbiased estimator. The covariance matrix is then
$$
\begin{aligned}
    \E \left[(\mathbf{c}-\hat{\mathbf{c}})(\mathbf{c}-\hat{\mathbf{c}})^T\right] &= \E \left[(\mathbf{G G}^T)^{-1} \mathbf{G} \mathbf{w} \mathbf{w}^T \mathbf{G}^T (\mathbf{G G}^T)^{-1} \right] \\
    &= \sigma^2 (\mathbf{G G}^T)^{-1}
\end{aligned}
$$
and the MSE is derived from the trace of the covariance matrix
$$
\begin{aligned}
    \E \left[ \| \mathbf{c} - \hat{\mathbf{c}}   \|^2 \right] &= \sigma^2\text{tr} \left[(\mathbf{GG}^T)^{-1} \right]. \\ 
\end{aligned}
$$

Following the \gls{fim} in (3.31, \cite{kay1993fundamentals}), and given $\mathbf{s}~\sim~\mathcal{N}(\mathbf{G}^T\mathbf{c}, \sigma^2 \mathbf{I} )$ from   (\ref{eq:newrangemodel}). we have the \gls{fim} as
$$
\begin{aligned}
    \mathbf{I}({\mathbf{c}})&= \frac{1}{\sigma^2} [\frac{\partial \mathbf{G}^T\mathbf{c}}{\partial \mathbf{c}}]^T [\frac{\partial \mathbf{G}^T\mathbf{c}}{\partial \mathbf{c}}] \\
    &= \frac{1}{\sigma^2} \mathbf{GG}^T.
\end{aligned}
$$

Therefore, \gls{crlb} is the inverse of the \gls{fim} denoted as $\vect{C}(\vect{c}) = \mathbf{I}^{-1}(\mathbf{c}) = \sigma^2(\mat{GG}^T)^{-1} $, and $\E \left[ (\mathbf{c}-\hat{\mathbf{c}})(\mathbf{c}-\hat{\mathbf{c}})^T \right] = \vect{C}(\vect{c})$ establishing the efficiency of the estimator.

\subsection{Proof of Proposition \ref{proposition:crlb bold_n}}
\label{ap:4}
We define the calibration error of $\hat{\mathbf{n}}_S$ as $\mathbf{e}$, expressed by the equation
$$
\mathbf{e}=\mathbf{n}_S-\hat{\mathbf{n}}_S=\mathbf{n}_S-\frac{(\mathbf{G G}^T)^{-1} \mathbf{G} \mathbf{s}}{\|(\mathbf{G G}^T)^{-1} \mathbf{G} \mathbf{s}\|}.
$$

Utilizing the established   (\ref{eq:newrangemodel}), $\mathbf{s}=\mathbf{G}^T \mathbf{c}+\mathbf{w}$, the error $\mathbf{e}$ can be reformulated as
$$
\mathbf{e} =\mathbf{n}_S-\frac{c \mathbf{n}_S+ (\mathbf{G} \mathbf{G}^T  )^{-1} \mathbf{G} \mathbf{w}}{ \|c \mathbf{n}_S+ (\mathbf{G G}^T  )^{-1} \mathbf{G w}  \|} .
$$

The subsequent analysis for \gls{mse} and bias presents complexities due to the norm normalization involved. A Taylor series approximation is employed to asymptotically bound the error. Let $\vect{w}=\sigma \vect{\nu}$, where $\vect{\nu}$ is a vector of i.i.d Gaussian unit distribution elements, we define $g(\vect{w})$ and its first order approximation as
$$
\begin{aligned}
&g(\mathbf{w}) =\frac{1}{\|c \mathbf{n}_S+(\mathbf{G G}^T)^{-1} \mathbf{G} \mathbf{w}\|} \\
& = \frac{1}{c} + \frac{\mathrm{d} g}{\mathrm{d} \mathbf{w}^T}\mathbf{w} + O(\sigma^2), ~~~ \text{as } \sigma^2 \rightarrow 0  \\
\end{aligned}
$$
where the derivative of $g$ with respect to $\vect{w}^T$ is 
$$
\begin{aligned}
\frac{\mathrm{d} g}{\mathrm{d} \mathbf{w}^T} & = -\frac{\left(c \mathbf{n}_S + (\mat{GG}^T)^{-1}\mat{G}\vect{w} \right)^T}{\left\|c \mathbf{n}_S + (\mat{GG}^T)^{-1}\mat{G}\vect{w}\right\|^3} (\mathbf{G G}^T)^{-1} \mathbf{G} \\
& =-\frac{\mathbf{n}_S^T}{c^2} (\mathbf{G} \mathbf{G}^T)^{-1} \mathbf{G}, ~~~\text{as } \sigma^2 \rightarrow 0.
\end{aligned}
$$

Finally, the approximation of $g(\mathbf{w}) $ is 
$$
g(\vect{w}) = \frac{1}{c} -\frac{\mathbf{n}_S^T}{c^2}(\mathbf{G G}^T)^{-1} \mathbf{Gw}  + O(\sigma^2), ~\text{as } \sigma^2 \rightarrow 0.
$$

Substituting this in the error $\mathbf{e}$ produces
$$
\begin{aligned}
\mathbf{e}&=\mathbf{n}_S-\left(c \mathbf{n}_S+(\mathbf{G G}^T)^{-1} \mathbf{G w}\right)  g(\vect{w}) \\
&=\frac{1}{c} \mathbf{n}_S \mathbf{n}_S^T(\mathbf{G G}^T)^{-1} \mathbf{G} \mathbf{w}  -\frac{1}{c}(\mathbf{G G}^T)^{-1} \mathbf{G w}  + \epsilon \\
& ~~~ +  \frac{1}{c^2}(\mathbf{G G}^T)^{-1} \mathbf{G} \mathbf{w} \mathbf{w}^T \mathbf{G}^T(\mathbf{G} \mathbf{G}^T)^{-1} \mathbf{n}_S, ~\text{as } \sigma^2 \rightarrow 0
\end{aligned}
$$
where $\epsilon$ represents a small error and $||\epsilon||_2 = O(\sigma)^2$ when $\sigma^2 \rightarrow 0$. This expectation simplifies further based on the assumption that $\E \left[ \mathbf{w}\mathbf{w}^T \right] = \sigma^2 \mathbf{I} $. Therefore, the expectation of bias $\mathbf{n}_S - \hat{\mathbf{n}}_S$ is 
$$
\E \left[ \mathbf{n}_S - \hat{\mathbf{n}}_S \right] = \frac{ O(\sigma^2) }{c^2} (\mathbf{GG}^T)^{-1} \mathbf{n}_S, ~\text{as } \sigma^2 \rightarrow 0 .
$$

The expected bias norm is
$$
 \E [\| \mathbf{n}_S - \hat{\mathbf{n}}_S \|_2 ] = \frac{O(\sigma^2)}{c^2} \sqrt{ \mathbf{n}_S^T (\mat{G}\mat{G}^T)^{-2} \mathbf{n}_S  }, ~\text{as } \sigma^2 \rightarrow 0 .
$$

\subsection{Proof of Proposition \ref{proposition:crlb c}}
\label{ap:3}
According to the   (\ref{eq:answer}) and   (\ref{eq:newrangemodel}) the gain error is 
$$
\begin{aligned}
    c-\hat{c}&=c-\|(\mathbf{G G}^T)^{-1} \mathbf{G} \mathbf{s}\| \\
    &=c- \|(\mathbf{G G}^T)^{-1} \mathbf{G}(\mathbf{G}^T\mathbf{c} + \vect{w} )\|.
\end{aligned}
$$

We also introduce a Taylor series approximation. Denoting $g(\vect{w}) = \norm{(\mathbf{G G}^T)^{-1} \mathbf{G}(\mathbf{G}^T\mathbf{c} + \mathbf{w})}$, and the derivative of $g$ in terms of $\vect{w}$ is 
$$
\begin{aligned}
    &g(\vect{w}) = c + \frac{\mathrm{d} g}{\mathrm{d} \mathbf{w}^T} \vect{w} + O(\sigma^2), ~~~~\text{as } \sigma^2 \rightarrow 0  \\
    &= c + \frac{ ((\mathbf{G G}^T)^{-1} \mathbf{G}\mathbf{G}^T\mathbf{c})^T (\mathbf{GG}^T)^{-1} \mathbf{G} } {\norm{(\mathbf{G G}^T)^{-1} \mathbf{G}\mathbf{G}^T\mathbf{c} }} \vect{w} \\ 
    &+  O(\sigma^2), ~~~\text{as } \sigma^2 \rightarrow 0.
\end{aligned}
$$

Incorporating the definition $\vect{w} = \sigma \vect{\nu}$,  as the $\sigma^2 \rightarrow 0$, the expected bias $\E[c-\hat{c}]$  asymptotically converges to
$$
\E \left[c - \hat{c}\right]= \E\left[-\sigma \vect{n}_S^T (\mathbf{GG}^T)^{-1} \mat{G} \vect{\nu} +  O(\sigma^2)\right] = O(\sigma^2).
$$
Finally, the expected bias norm is
$$
\E \left[ || c - \hat{c}||_2\right] = O(\sigma^2).
$$

\subsection{Proof of Theorem \ref{theorem:optimal_calibration_dataset}}
\label{ap:optimal_calibrationdata}
This section provides a proof of the theorem, structured through several steps. Initially, Lemma~\ref{lemma:ggt_lower_bound} establishes the  lower bound of cost function    (\ref{eq:optimaldataset}), showing it is achieved when $\mat{GG}^T$ is diagonal. Next, Lemma~\ref{lemma:ggt_diagonal} asserts the first two diagonal elements of $\mat{GG}^T$ must be equal. We conclude by demonstrating our scheme meets these conditions and achieves the lower bound, indicating its effectiveness.

\begin{lemma}
\label{lemma:ggt_lower_bound}
For the cost function   (\ref{eq:optimaldataset}), the lower bound is attached when $\mat{GG}^T$ is a diagonal matrix.
\end{lemma}
\begin{proof}
    Please see Appendix.~\ref{ap:ggt_lower_bound}.
\end{proof}

\begin{lemma}
\label{lemma:ggt_diagonal}
The first two diagonal elements of $\mat{GG}^T$ should be equal.
\end{lemma}
\begin{proof}
    Please see Appendix.~\ref{ap:ggt_diagonal}.
\end{proof}

Aligned with Lemma~\ref{lemma:ggt_lower_bound} and  Lemma~\ref{lemma:ggt_diagonal}, our goal is  to construct a matrix $\mat{G}$ such that its rows are mutually orthogonal and the first two diagonal elements of $\mat{GG}^T$ are equal. An optimally effective strategy involves positioning the points along the circumference of a circle with radius $r^*$ from   (\ref{eq:r*}). Recall the definition of each column vector in $\mathbf{G}$ , $\mathbf{g}_n=\frac{h}{d_n^4}\left(\mathbf{r}_{R, n}-\mathbf{r}_S\right)$, where $d_n^2=\left\|\mathbf{r}_{R, n}-\mathbf{r}_S\right\|^2=r^{*2}+h^2$. Thus, the matrix $\mathbf{G}\in \mathbb{R}^{3 \times N}$ is composed as
$$
\mathbf{G}=\frac{h}{\left(r^{*2}+h^2\right)^2}\left[\begin{array}{c}
r^* \cos \left(\frac{2 \pi}{N} 0\right), \cdots, r^* \cos \left(\frac{2 \pi}{N} n\right) \\
r^* \sin \left(\frac{2 \pi}{N} 0\right), \cdots, r^* \sin \left(\frac{2 \pi}{N} n\right) \\
-h, \cdots,-h
\end{array}\right]  .
$$

The orthogonality of $\mathbf{G}$ first and second rows to the third is affirmed by showing that their respective dot products with the third row sum to zero
$$
\sum_n r^{*2} \cos \left(\frac{2 \pi}{N} n\right)(-h)=\sum_n r^{*2} \sin \left(\frac{2 \pi}{N} n\right)(-h)=0.
$$

Furthermore, the first and second rows are orthogonal by virtue of the product-to-sum formula, as indicated
$$
\begin{aligned}
& \sum_n r^{*2} \cos \left(\frac{2 \pi}{N} n\right) \sin \left(\frac{2 \pi}{N} n\right) \\
& =r^{*2} \frac{1}{2} \sum_n\left(\sin (0)+\sin \left(2 \frac{2 \pi}{N} n\right)\right) \\
& =  r^{*2} \frac{1}{2} \Im\left(\frac{1-e^{\jmath  \frac{4 \pi}{N} N}}{ 1-e^{ \jmath \frac{4 \pi}{N}}}\right)=0
\end{aligned}
$$
where the $\Im(\cdot)$ denotes the imaginary part. 

Additionally, the squared norms of each row align with the trace requirements of $\mathbf{GG}^T$ without the factor $\frac{h}{(r^{*2}+h^2)^2}$, calculated as follows
$$
\left\{\begin{array}{l}
\sum_{n=0}^{N-1} r^{*2} \cos^2 \left(\frac{2 \pi}{N} n\right)=r^{*2} \sum_{n=0}^{N-1} \frac{1+\cos \left(\frac{4 \pi}{N} n\right)}{2}=r^{*2} \frac{N}{2} \\
\sum_{n=0}^{N-1} r^{*2} \sin ^2\left(\frac{2 \pi}{N} n\right)=r^{*2} \sum_{n=0}^{N-1} \frac{1-\cos \left(\frac{4 \pi}{N} n\right)}{2}=r^{*2} \frac{N}{2} \\
\sum_{n=0}^{N-1} h^2=N h^2.
\end{array}\right.
$$

These squared norms ensure that the trace of $\mat{GG}^T$ achieves the lower bound.

\subsection{Proof of Lemma \ref{lemma:ggt_lower_bound}}
\label{ap:ggt_lower_bound}
Leveraging inequality (3.12, 65 in \cite{kay1993fundamentals}) for a positive definite matrix, we have
$$
[(\mat{GG}^T)^{-1}]_{i,i} \geq \frac{1}{[\mat{GG}^T]_{i,i}}
$$
where $[\mathbf{GG}^T]_{i,i}$ represents the $i$-th diagonal element of $\mathbf{GG}^T$. This leads to the inequality
$$
    \operatorname{tr}[(\mathbf{GG}^T)^{-1}] \geq \sum_{i} \frac{1}{[\mathbf{GG}^T]_{i,i}}.
$$

Equality is attained either when $\mathbf{GG}^T$ is a diagonal matrix or when its diagonal elements are its eigenvalues. For our calibration scheme, we opt to make $\mat{GG}^T$ diagonal, as aligning its diagonal elements with the eigenvalues is challenging if it's not already a diagonal matrix.

\subsection{Proof of Lemma \ref{lemma:ggt_diagonal}}
\label{ap:ggt_diagonal}
In the calibration of each LED, which is treated individually, the LED under calibration is assumed to be the origin, denoted as $\vect{r}_S = [0,0,h]^T$. We define standard basic vectors in $\mathbb{R}^3$, $\vect{e}_1 =[1,0,0]^T$, $\vect{e}_2 =[0,1,0]^T$ and $\vect{e}_3 =[0,0,1]^T$. Thus $\mathbf{r}_{R,n} - \mathbf{r}_S = x_n \mathbf{e}_1 + y_n \mathbf{e}_2 - h\mathbf{e}_3$ and express $d_n^2= x_n^2 + y_n^2 + h^2$. Given $\mathbf{g}_n = \frac{h}{d_n^4} (\mathbf{r}_{R, n}-\mathbf{r}_S) \in \mathbb{R}^{3 \times 1}$, we can denote $\mathbf{GG}^T$ as follows
$$
\begin{aligned}
\mathbf{G} \mathbf{G}^T 
&=h^2 \sum_{n=0}^{N-1} \frac{1}{\left(x_n^2+y_n^2+h^2\right)^4} \\
& \times \left(x_n \mathbf{e}_1+y_n \mathbf{e}_2-h \mathbf{e}_3\right)\left(x_n \mathbf{e}_1+y_n \mathbf{e}_2-h \mathbf{e}_3\right)^T.
\end{aligned}
$$

Furthermore, expressing the vectors in polar coordinates yields $\mathbf{r}_{R, n}-\mathbf{r}_S=r_n \cos \left(\phi_n\right) \mathbf{e}_1+r_n \sin \left(\phi_n\right) \mathbf{e}_2-h \mathbf{e}_3$. Here, the notations are implicitly defined as functions of $x_n$ and $y_n$. Consequently, $d_n^2=r_n^2+h^2$. This leads us to express $\mathbf{GG}^T$ in polar coordinates as follows
$$
\begin{aligned}
\mathbf{G} \mathbf{G}^T 
&= h^2 \sum_{n=0}^{N-1} \frac{1}{\left(r_n^2+h^2\right)^4}  \\
&\times \left(r_n \cos \left(\phi_n\right) \mathbf{e}_1+r_n \sin \left(\phi_n\right) \mathbf{e}_2-h \mathbf{e}_3\right)  \\ 
& \times \left(r_n \cos \left(\phi_n\right) \mathbf{e}_1+r_n \sin \left(\phi_n\right) \mathbf{e}_2-h \mathbf{e}_3\right)^T.
\end{aligned}
$$

As established in Lemma~\ref{lemma:ggt_lower_bound} and considering the cost function in   (\ref{eq:optimaldataset}), we have effectively established the lower bound for this objective function. This bound can be expressed as
$$
\begin{aligned}
&\sum_{i=0}^2 \frac{1}{\left[\mathbf{G G}^T\right]_{i i}} =\frac{1}{h^2}\left[\left(\sum_{n=0}^{N-1} \frac{r_n^2 \cos ^2 \phi_n}{\left(r_n^2+h^2\right)^4}\right)^{-1}\right. \\
& +\left(\sum_{n=0}^{N-1} \frac{r_n^2 \sin ^2 \phi_n}{\left(r_n^2+h^2\right)^4}\right)^{-1}  \left.+\left(\sum_{n=0}^{N-1} \frac{h^2}{\left(r_n^2+h^2\right)^4}\right)^{-1}\right].
\end{aligned}
$$

Defining $\eta_n=r_n^2 / h^2$, a positive scalar, we reformulate the task of minimizing the previous expression as
$$
\begin{aligned}
f &= \left(\sum_{n=0}^{N-1} \frac{\eta_n \cos ^2 \phi_n}{\left(\eta_n+1\right)^4}\right)^{-1} +\left(\sum_{n=0}^{N-1} \frac{\eta_n \sin ^2 \phi_n}{\left(\eta_n+1\right)^4}\right)^{-1} \\
&+ \left(\sum_{n=0}^{N-1} \frac{1}{\left(\eta_n+1\right)^4}\right)^{-1}.
\end{aligned}
$$

Considering the relationship $\sin ^2 \phi_n=1-\cos ^2 \phi_n \in[0,1]$, an increase in $\cos ^2 \phi_n$ corresponds to a decrease in $\sin ^2 \phi_n$, and vice versa. With the third term solely dependent on $\eta_n$, and assuming $\eta_n$ are fixed, we define $\alpha_n=\frac{\eta_n}{\left(\eta_n+1\right)^4}$ and $\beta_n=$ $\cos ^2 \phi_n \in[0,1]$. The optimization with respect to $\phi_n$ can then be rewritten as
$$
\min _{\beta_n \in[0,1]} J=\left(\sum_n^{N-1} \alpha_n \beta_n\right)^{-1}+\left(\sum_n^{N-1} \alpha_n\left(1-\beta_n\right)\right)^{-1}.
$$

Taking derivatives and setting them to zero yields
$$
\begin{aligned}
\frac{\mathrm{d} J}{\mathrm{d} \beta_m} & =0 \\
\left(\sum_{n=0}^{N-1} \alpha_n \beta_n\right)^{-2} \alpha_m &= \left(\sum_{n=0}^{N-1} \alpha_n\left(1-\beta_n\right)\right)^{-2} \alpha_m \\
\sum_{n=0}^{N-1} \alpha_n \beta_n & =\sum_{n=0}^{N-1} \alpha_n\left(1-\beta_n\right)
\end{aligned}
$$

where we used the fact that all quantities are positive. Hence this implies that 
$$
\begin{aligned}
\left(\sum_n^{N-1} \frac{\eta_n \cos ^2 \phi_n}{\left(\eta_n+1\right)^4}\right)^{-1} & =\left(\sum_n^{N-1} \frac{\eta_n \sin ^2 \phi_n}{\left(\eta_n+1\right)^4}\right)^{-1} \\
\sum_n^{N-1} \frac{\eta_n \cos ^2 \phi_n}{\left(\eta_n+1\right)^4} & =\sum_n^{N-1} \frac{\eta_n \sin ^2 \phi_n}{\left(\eta_n+1\right)^4} \\
\sum_n^{N-1} \frac{\eta_n \cos ^2 \phi_n}{\left(\eta_n+1\right)^4} & =\frac{1}{2} \sum_n^{N-1} \frac{\eta_n}{\left(\eta_n+1\right)^4} 
\end{aligned}
$$

we observe that the dependence on $\phi_n$ is eliminated. Therefore, our objective is to minimize
$$
f=4\left(\sum_{n=0}^{N-1} \frac{\eta_n}{\left(\eta_n+1\right)^4}\right)^{-1}+\left(\sum_{n=0}^{N-1} \frac{1}{\left(\eta_n+1\right)^4}\right)^{-1}.
$$

To find the optimal values of $\eta_n$, we take the derivative of $f$ with respect to $\eta_m$
$$
\begin{aligned}
\frac{\mathrm{d} f}{\mathrm{d} \eta_m} & =-4\left(\sum_{n=0}^{N-1} \frac{\eta_n}{\left(\eta_n+1\right)^4}\right)^{-2} \frac{\left(\eta_m+1\right)^4-4 \eta_m\left(\eta_m+1\right)^3}{\left(\eta_m+1\right)^8} \\
& -\left(\sum_{n=0}^{N-1} \frac{1}{\left(\eta_n+1\right)^4}\right)^{-2} \frac{-4\left(\eta_m+1\right)^3}{\left(\eta_m+1\right)^8}
\end{aligned}
$$
Setting this derivative to zero yields
$$
\begin{aligned}
\Bigg(\sum_n^{N-1} \frac{1}{(\eta_n+1)^4} \Bigg)^{-2} &=\Bigg(\sum_n^{N-1} \frac{\eta_n}{(\eta_n+1)^4} \Bigg)^{-2}(1-3\eta_m) \\
\left(\frac{\sum_n^{N-1} \frac{\eta_n}{\left(\eta_n+1\right)^4}}{\sum_n^{N-1} \frac{1}{\left(\eta_n+1\right)^4}}\right)^2&=1-3 \eta_m .
\end{aligned}
$$

We observe that the left-hand side of the equation is a positive constant independent of $m$, while the right-hand side represents a linear curve with a negative slope. The intersection of a constant with a linear curve yields only one solution, indicating that $\eta_m$ is at the optimum equal for all $m$ and we write $\eta_m=\eta$. We need to solve
$$
\begin{aligned}
\left(\frac{N \frac{\eta}{(\eta+1)^4}}{N \frac{1}{(\eta+1)^4}}\right)^2 & =-3 \eta+1 \\
0 & =\eta^2+3 \eta-1 \\
\eta & =\frac{-3 \pm \sqrt{13}}{2} .
\end{aligned}
$$

Given the positive nature of $\eta$, the radius is determined as 
\begin{equation}
\label{eq:r*}
    r^*=\sqrt{\frac{-3+\sqrt{13}}{2}} h \approx 0.55 h .
\end{equation}

This leads to the conclusion that
$$
\operatorname{tr}\left[\left(\mathbf{G G}^T\right)^{-1}\right] \geq \frac{1}{N} \frac{\left(h^2+r^{*2}\right)^4}{h^2}\left(\frac{4}{r^{*2}}+\frac{1}{h^2}\right).
$$

To satisfy the bound, a configuration is needed where the product $\mathbf{G G}^T$ is as follows
$$
\mathbf{G G}^T=\frac{ h^2}{\left(h^2+r^{*2}\right)^4} \operatorname{diag}\left[\frac{Nr^{*2}}{2}, \frac{Nr^{*2}}{2}, Nh^2\right].
$$


\bibliographystyle{IEEEtran}
\bibliography{IEEEabrv,mybibfile}

\end{document}